\documentclass[sigconf]{acmart}
\usepackage{ifthen}
\usepackage{fixltx2e}
\usepackage{pifont}
\newcommand{\cmark}{\text{\ding{51}}}%
\newcommand{\xmark}{\text{\ding{55}}}%
\DeclareMathAlphabet{\mathpzc}{OT1}{pzc}{m}{it}
\usepackage{amssymb} 
\usepackage{todonotes} 
\usepackage{multirow}

\usepackage{diagrams}
\usepackage{color}
\usepackage{graphicx}
\usepackage{color}
\usepackage{verbatim}
\usepackage{graphicx}
\usepackage{eurosym}
\usepackage{bm}
\usepackage{array}
\usepackage{float}
\usepackage{ stmaryrd }
\usepackage{braket}
\usepackage{tikz}
\usepackage{tikz-cd}
\usepackage{mathtools}
\usepackage{listings}

\newcommand{\nat}{{\omega}}

\newcommand{\Real}{{\mathbb{R}}}

\newcommand{\AsRule}[1]{{{\textsc{#1}}}}
\newcommand{\AsFunction}[1]{{\mathtt{#1}}}

\newcommand{\Pack}{\AsFunction{{pack}}}
\newcommand{\Lvl}{\AsFunction{{lvl}}}

\MakeRobust{\Pack}

\newcommand{\SomeElement}[1]{
	\mathtt{
	\ifthenelse{\equal{#1}{1}}
		{{x}}
		{\ifthenelse{\equal{#1}{2}}
		{{y}}
		{\ifthenelse{\equal{#1}{3}}
		{{z}}
		{{{#1}}}}}
	}
}
\MakeRobust{\SomeElement}

\newcommand{\SomeState}[1]{
	\mathtt{
	\ifthenelse{\equal{#1}{1}}
		{{s}}
		{\ifthenelse{\equal{#1}{2}}
		{{s'}}
		{\ifthenelse{\equal{#1}{3}}
		{{s''}}
		{{{#1}}}}}
	}
}
\MakeRobust{\SomeState}

\newcommand{\SomeSet}[1]{
	{\ifthenelse{\equal{#1}{1}}
		{X}
		{\ifthenelse{\equal{#1}{2}}
		{Y}
		{\ifthenelse{\equal{#1}{3}}
		{Z}
		{{#1}}}}
	}
}
\MakeRobust{\SomeSet}

\newcommand{\SomePredicate}[1]{
	{\ifthenelse{\equal{#1}{1}}
		{{\AsRule{P}}}
		{\ifthenelse{\equal{#1}{2}}
		{{\AsRule{Q}}}
		{\ifthenelse{\equal{#1}{3}}
		{{\AsRule{R}}}
		{{#1}}}}
	}
}
\MakeRobust{\SomePredicate}

\newcommand{\SomeRelation}[1]{
	{\ifthenelse{\equal{#1}{1}}
		{{R}}
		{\ifthenelse{\equal{#1}{2}}
		{{S}}
		{\ifthenelse{\equal{#1}{3}}
		{{T}}
		{{#1}}}}
	}
}
\MakeRobust{\SomeRelation}

\newcommand{\SomeFunction}[1]{
	{\ifthenelse{\equal{#1}{1}}
		{f}
		{\ifthenelse{\equal{#1}{2}}
		{g}
		{\ifthenelse{\equal{#1}{3}}
		{h}
		{{#1}}}}
	}
}
\MakeRobust{\SomeFunction}

\newcommand{\SomeCategory}[1]{
	{\ifthenelse{\equal{#1}{1}}
		{{\mathbf{C}}}
		{\ifthenelse{\equal{#1}{2}}
		{{\mathbf{D}}}
		{\ifthenelse{\equal{#1}{3}}
		{{\mathbf{E}}}
		{{\mathbf{#1}}}}}
	}
}
\MakeRobust{\SomeCategory}

\newcommand{\AsFunctor}[1]{\mathpzc{#1}}
\newcommand{\SomeFunctor}[1]{
	{
	\ifthenelse{\equal{#1}{1}}
		{\AsFunctor{F}}
		{\ifthenelse{\equal{#1}{2}}
		{\AsFunctor{G}}
		{\ifthenelse{\equal{#1}{3}}
		{\AsFunctor{H}}
		{\AsFunctor{{#1}}}}}
	}
}
\MakeRobust{\SomeFunctor}

\newcommand{\AsAlgebra}[1]{
	{\mathfrak{{#1}}}
}
\newcommand{\SomeAlgebra}[1]
{
	{\ifthenelse{\equal{#1}{1}}
	{\AsAlgebra{A}}
	{\ifthenelse{\equal{#1}{2}}
		{\AsAlgebra{B}}
		{\ifthenelse{\equal{#1}{3}}
			{\AsAlgebra{C}}
			{\PackageWarning{Preamble}{non-standard algebra symbol}#1}
		}
	}
	}
}
\MakeRobust{\SomeAlgebra}

\newcommand{\AsCoalgebra}[1]{
	{\mathbb{{#1}}}
}
\newcommand{\SomeCoalgebra}[1]{
	{\AsCoalgebra{\MakeUppercase{\SomeSet{#1}}}}
}
\MakeRobust{\SomeCoalgebra}

\newcommand{\SomeInput}[1]
{
	{\ifthenelse{\equal{#1}{1}}
	{\mathtt{a}}
	{\ifthenelse{\equal{#1}{2}}
		{b}
		{\ifthenelse{\equal{#1}{3}}
			{c}
			{\PackageWarning{Preamble}{non-standard input symbol}#1}
		}
	}
	}
}
\MakeRobust{\SomeInput}

\newcommand{\SomeSetOfInputs}[1]
{
	\ifthenelse{\equal{#1}{1}}
	{{A}}
	{\ifthenelse{\equal{#1}{2}}
		{B}
		{\ifthenelse{\equal{#1}{3}}
			{C}
			{\PackageWarning{Preamble}{non-standard set of inputs symbol}#1}
		}
	}
}
\MakeRobust{\SomeSetOfInputs}

\newcommand{\High}{\mathcal{H}}
\newcommand{\Low}{\mathcal{L}}

\newcommand{\SomeSequenceOfInputs}[1]{
{
	{\ifthenelse{\equal{#1}{1}}
	{w}
	{\ifthenelse{\equal{#1}{2}}
		{u}
		{\ifthenelse{\equal{#1}{3}}
			{v}
			{\PackageWarning{Preamble}{non-standard sequence of inputs symbol}#1}
		}
	}
}}
}
\MakeRobust{\SomeSequenceOfInputs}

\newcommand{\SomeSetOfSequencesOfInputs}[1]{\SomeSetOfInputs{#1}^{*}}
\MakeRobust{\SomeSetOfSequencesOfInputs}

\newcommand{\SomeStreamOfInputs}[1]{
	\ifthenelse{\equal{#1}{1}}
		{\sigma}
		{\ifthenelse{\equal{#1}{2}}
		{\tau}
		{\ifthenelse{\equal{#1}{3}}
		{\rho}
		{{#1}}}}
	}
\MakeRobust{\SomeStreamOfInputs}

\newcommand{\SomeSetOfStreamsOfInputs}[1]{\SomeSetOfInputs{#1}^{\nat}}
\MakeRobust{\SomeSetOfStreamsOfInputs}

\newcommand{\SomeOutput}[1]
{
	\ifthenelse{\equal{#1}{1}}
	{b}
	{\ifthenelse{\equal{#1}{2}}
		{p}
		{\ifthenelse{\equal{#1}{3}}
			{q}
			{\PackageWarning{Preamble}{non-standard output symbol}#1}
		}
	}
}
\MakeRobust{\SomeOutput}

\newcommand{\SomeSetOfOutputs}[1]
{
	\ifthenelse{\equal{#1}{1}}
	{{B}}
	{\ifthenelse{\equal{#1}{2}}
		{P}
		{\ifthenelse{\equal{#1}{3}}
			{Q}
			{\PackageWarning{Preamble}{non-standard set of outputs symbol}#1}
		}
	}
}
\MakeRobust{\SomeSetOfOutputs}

\newcommand{\SomeSequenceOfOutputs}[1]
{
	\ifthenelse{\equal{#1}{1}}
	{\underline{w}}
	{\ifthenelse{\equal{#1}{2}}
		{\underline{u}}
		{\ifthenelse{\equal{#1}{3}}
			{\underline{v}}
			{\PackageWarning{Preamble}{non-standard sequence of outputs symbol} \underline{#1}}
		}
	}
}
\MakeRobust{\SomeSequenceOfOutputs}

\newcommand{\SomeSetOfSequencesOfOutputs}[1]{\SomeSetOfOutputs{#1}^{*}}
\MakeRobust{\SomeSetOfSequencesOfOutputs}

\newcommand{\SomeException}[1]
{
	\ifthenelse{\equal{#1}{1}}
	{e}
	{\ifthenelse{\equal{#1}{2}}
		{d}
		{\ifthenelse{\equal{#1}{3}}
			{c}
			{\PackageWarning{Preamble}{non-standard Exception symbol}#1}
		}
	}
}
\MakeRobust{\SomeException}

\newcommand{\SomeSetOfExceptions}[1]
{
	\ifthenelse{\equal{#1}{1}}
	{E}
	{\ifthenelse{\equal{#1}{2}}
		{F}
		{\ifthenelse{\equal{#1}{3}}
			{G}
			{\PackageWarning{Preamble}{non-standard set of exceptions symbol}#1}
		}
	}
}
\MakeRobust{\SomeSetOfExceptions}

\newcommand{\AsSemiring}[1]{\AsCoalgebra{#1}}
\newcommand{\SomeSemiring}[1]{\AsSemiring{\MakeUppercase{\SomeSet{#1}}}}
\MakeRobust{\SomeSemiring}

\newcommand{\SomeFPS}[1]{
	\ifthenelse{\equal{#1}{1}}
		{\sigma}
		{\ifthenelse{\equal{#1}{2}}
		{\gamma}
		{\ifthenelse{\equal{#1}{3}}
		{\rho}
		{{#1}}}}
	}
\MakeRobust{\SomeFPS}

\newcommand{\SomeBehaviour}[1]{{\SomeFPS{#1}}}
\MakeRobust{\SomeBehaviour}

\newcommand{\TheSystem}{{\TheBehaviourOf{\cdot}}}
\newcommand{\TheBehaviourOf}[1]{{{\llbracket#1\rrbracket}}}

\newcommand{\AsSequence}[1]{{{\left(#1\right)}}}
\newcommand{\AsTuple}[1]{{{\left(#1\right)}}}

\usetikzlibrary{arrows,automata}

\makeatletter
\newcommand{\da@rightarrow}{\mathchar"0\hexnumber@\symAMSa 4B }
\newcommand{\da@leftarrow}{\mathchar"0\hexnumber@\symAMSa 4C }
\newcommand{\xdashedrightarrow}[2][]{%
  \mathrel{%
    \mathpalette{\da@xarrow{#1}{#2}{}\da@rightarrow{\,}{}}{}%
  }%
}
\newcommand{\xdashedleftarrow}[2][]{%
  \mathrel{%
    \mathpalette{\da@xarrow{#1}{#2}\da@leftarrow{}{}{\,}}{}%
  }%
}
\newcommand{\da@xarrow}[7]{%
  \sbox0{$\ifx#7\scriptstyle\scriptscriptstyle\else\scriptstyle\fi#5#1#6\m@th$}%
  \sbox2{$\ifx#7\scriptstyle\scriptscriptstyle\else\scriptstyle\fi#5#2#6\m@th$}%
  \sbox4{$#7\dabar@\m@th$}%
  \dimen@=\wd0 %
  \ifdim\wd2 >\dimen@
    \dimen@=\wd2 %
  \fi
  \count@=2 %
  \def\da@bars{\dabar@\dabar@}%
  \@whiledim\count@\wd4<\dimen@\do{%
    \advance\count@\@ne
    \expandafter\def\expandafter\da@bars\expandafter{%
      \da@bars
      \dabar@ 
    }%
  }%
  \mathrel{#3}%
  \mathrel{%
    \mathop{\da@bars}\limits
    \ifx\\#1\\%
    \else
      _{\copy0}%
    \fi
    \ifx\\#2\\%
    \else
      ^{\copy2}%
    \fi
  }%
  \mathrel{#4}%
}

\makeatletter
\providecommand*{\twoheadrightarrowfill@}{%
  \arrowfill@\relbar\relbar\twoheadrightarrow
}
\providecommand*{\twoheadleftarrowfill@}{%
  \arrowfill@\twoheadleftarrow\relbar\relbar
}
\providecommand*{\xtwoheadrightarrow}[2][]{%
  \ext@arrow 0579\twoheadrightarrowfill@{#1}{#2}%
}
\providecommand*{\xtwoheadleftarrow}[2][]{%
  \ext@arrow 5097\twoheadleftarrowfill@{#1}{#2}%
}
\makeatother

\usepackage{wrapfig}

\setcopyright{none} 
\acmYear{2017}

    \author{Eric Rothstein Morris, Carlos G. Murguia, Mart\'in Ochoa}
\affiliation{
  \institution{Singapore University of Technology and Design}            
}
 \email{eric_rothstein,murguia_rendon,martin_ochoa@sutd.edu.sg}         
 
\newcommand{\Nat}{{\mathbb {N}}}
\begin{document}
\title{Design-Time Quantification of Integrity in Cyber-Physical-Systems } 

\begin{abstract}
In a software system it is possible to quantify the amount of information that 
is leaked or corrupted by analysing the flows of information present in the 
source code. In a cyber-physical system, information flows are not only present 
at the digital level, but also at a physical level, and to and fro the two 
levels. 
In this work, we provide a methodology to formally analyse a 
Cyber-Physical System composite model (combining physics and control) using an 
information flow-theoretic approach. We use this approach to quantify the level 
of vulnerability of a system with respect to attackers with different 
capabilities. We illustrate our approach by means of a water distribution case 
study. 
\end{abstract}



\maketitle


\section{Introduction}
A \emph{cyber-physical system} (CPS) is a system that intertwines components from the physical and digital worlds. Some examples of CPSs include: cars, aircrafts, water treatment plants, industrial control systems, and critical infrastructures. Security violations in a safety-critical CPS has notable effects in the physical world, \emph{e.g.,} in late 2007 or early 2008, the Stuxnet attack against an Iranian control system allegedly sabotaged centrifuges in uranium enrichment plants, causing them to rapidly deteriorate \cite{StuxnetWeb,Stuxnet}, and in 2014, hackers struck a steel mill in Germany and disrupted the control system, which prevented a blast furnace from properly shutting down, causing massive damage to the facility \cite{WiredArticle,Lagebericht2014}.

Although in the scientific literature on information security \emph{confidentiality} has traditionally enjoyed more attention than \emph{integrity}, as observed for instance by Clark and Wilson \cite{ClarkWilson87}: ``in the commercial environment, preventing disclosure is often important, but preventing unauthorized data modification is usually paramount.''
This holds particularly true for many Industrial Control Systems (ICSs) at yet another level: their security priority is not the protection of confidential data, but the protection of their physical assets.  Gollmann and Krotofil reinforce this paradigm in \cite{CPSSec}, stating that the traditional CIA (Confidentiality-Integrity-Availability) triad should be reversed when studying the security of CPSs
. They argue that the enforcement of integrity in CPSs should not only consider the ``traditional approach for IT systems'', \emph{i.e.}, protecting \emph{component logic} and \emph{communication}, but that we must also protect the integrity of {observations}; more precisely, we have to protect the \emph{veracity} of sensor data and check its \emph{plausibility}.

{Traditionally, the enforcement of CPS integrity has focused on guaranteeing that systems perform as their designers intended. This is usually evaluated using \emph{control-theoretic methodologies}, specifically \emph{fault-detection techniques}. Control theory researchers model attacks to CPSs as time-series with specific structures affecting sensor measurements and/or control signals. Depending on their capabilities, attackers have the power to control when, where, and how attacks are induced to the system. Research on attack detection and mitigation has mainly focused on the so-called \emph{integrity attacks}, i.e., attacks that put at risk the proper operation and physical integrity of CPSs. Integrity attacks include stealthy attacks \cite{CPSStealthAttacks}, message replay \cite{CPSReplayAttacks}, covert attacks \cite{CPSCovertAttacks}, and false-data injection \cite{CPSDataInjectionAttacks}, among others. Due to their focus on fault-detection techniques, many results based on control theory aim to protect CPSs from integrity attacks by relying on {monitoring} the physics of the system to detect anomalies (see, e.g., \cite{CPSInvariantsForDetection,LimitingImpactStealthyAttacks,CPSAttackDetection,CPSAttacksAgainstPCS,CPSIntegrityAttacks,CPSDetectingIntegrityAttacksScada}).
 }

{
In sum, most control-theoretical approaches to CPS security assume the
existence of a probabilistic behavioural model that is used as a
``white-list'' of normal behaviour (attack-free), and their goal is to monitor and protect the integrity of the system with respect to this ideal model at runtime, when attacks might arise. There are some limitations with this approach. On the one hand, most of the work following this approach is reactive, and as such it does not shed light on how to improve a given CPSs design in order to make it more resilient to attacks (although there are some recent results on \emph{redesigning controllers and models} to improve robustness of CPSs against attacks, e.g., \cite{Carlos_Justin3,Weerakkody}). On the other hand, an inherent limitation of many behavioural models for CPS is that they usually are approximate due to linearisation, and might produce a high false positive rate when used in practice; also, many works (e.g., \cite{LimitingImpactStealthyAttacks,CPSDetectingIntegrityAttacksScada,IFCPSSec}) rely on the assumption that it is possible to determine whether the system is operating normally or under attack, but there could be non-malicious deviations of a given behavioural model due to routine maintenance operations or other random factors.
Moreover, only few approaches (such as \cite{IFCPSSec,Gupta2,Carlos_Justin3}) quantify the severity and consequences of attacks on the system.
}

Information flow analysis (IFA) has traditionally focused on determining whether sensitive/secret information flows to where it is not intended to. There is a large body of literature on the application of IFA for the attestation of {data confidentiality}, including \cite{BellLapadula,SecureInformationFlows,QIF}, among others. Although it is well known that conceptually integrity can be seen as the dual of confidentiality, arguably the attestation of integrity has not received much attention from the IFA community. We thus suggest, as an interesting application scenario of IFA, to develop an alternative and complementary approach to control theory to reason
about CPS security based on well-known computer science foundations for
integrity. We make
the fundamental observation that in practice we often want to protect the
integrity of the state of a (physical) \emph{process variable} (i.e. the level 
of water of a tank, the temperature of a certain device or material, the 
chemical concentration of a given medium etc.). We propose then to model a CPS 
as the composition of a cyber (digital) state-machine modelling the controller 
logic, and the corresponding expected physical world reaction. In this setting 
we can apply information-flow inspired definitions and techniques to formally 
assess and quantify the effects of an attacker controlling certain aspects of 
the system (sensors, actuators, or  communication channels) on the critical 
physical variables we intend to protect. 

Although our approach can shed light on the dependencies between
certain inputs and certain critical physical states of the system, we remark that our analysis
assumes some reasonable limitations on the considered attackers. In particular, we assume
that attackers control only a subset of sensors/actuators and that attackers do
not subvert the controller's logic. Clearly, if attackers control all sensors
and actuators or are able to replace the control logic (as in the case of
Stuxnet), then they can drive a CPS to any desired state (unless orthogonal
controls are in place). However, we believe that the analysis proposed in this work can be
useful to decide where and if to include redundancy in the number of sensors
and/or in the control logic, in order to make attacks more difficult to perform. 

{
Given that our
approach is intended to verify the security of CPSs at design-time, we make assumptions that might be
unrealistic when considered in an approach that monitors security at runtime (in particular, we assume that we know the state of the
controller and of the physical process). Nevertheless, we consider our assumptions to be realistic at design-time, given the existence of a models that
describe the dynamical systems (where states are explicitly described), and given the existence of a state machine that models the controller (where the current state is usually known).
}

\textbf{Problem statement:} To summarize, in this work we address the following
questions: \emph{Can we identify at design time which inputs to a CPS are most
harmful if controlled by an attacker? Can we improve the design of the CPS
at design time and justify our models formally with respect to a precise
integrity notion?}

\textbf{Approach:} We propose to model a CPS as a composition of the control
logic and the expected systems behaviour in terms of a finite and deterministic
state machine. We then formally quantify the impact of various attackers
(controlling physical or logical aspects of the system) to the integrity
of critical physical states by means of information flow analysis.

\textbf{Contributions:} We make the following contributions \textbf{a)} We present a logical framework to formally reason about integrity in CPS models,
reconciling information flow analysis and control theoretical aspects. \textbf{b)} We discuss how our approach can be used to spot and quantify harmful information flows in CPS. \textbf{c)} We illustrate the usefulness of
our approach by means of simple but realistic models concerning a water
distribution CPS and show that we can identify non-trivial integrity-harming flows
in a CPS.


{
}


\section{Control Theory Preliminaries}
Gollmann and Krotofil highlight in \cite{CPSSec} that ``to work on cyber-physical systems, one has to combine the devious mind of the security expert with the technical expertise of the control engineer.'' Unfortunately, this expertise is ``not commonly found in the IT security community,'' according to the authors of \cite{CPSSecVinyl}. Moreover, this lack of expertise is worsened by the {language barrier} between the disciplines of control theory and IT security. Thus, in this section, we present the model-based techniques for CPSs security broadly used by the systems and control community.
\subsection{Linear Time-Invariant Models.}
During the last decade, there has been an increasing tendency to use physics-based models of CPSs to detect and quantify the effect of attacks on the system performance \cite{CPSAttacksAgainstPCS,LimitingImpactStealthyAttacks,CPSDetectingIntegrityAttacksScada,CPSIntegrityAttacks,Carlos_Justin1,Carlos_Justin2,Carlos_Justin3}. These physics-based models focus on the normal operation of the CPS and work as prediction models that are used to confirm that control commands and measurements are valid and plausible. Often, dynamical models of physical systems are approximated around their operation points or approximated using input-output data. These lead to approximated models which are often linear and time-invariant.

Following the work in \cite{CPSAttacksAgainstPCS,LimitingImpactStealthyAttacks,CPSDetectingIntegrityAttacksScada,CPSIntegrityAttacks,Carlos_Justin1,Carlos_Justin2,Carlos_Justin3}, here, we only consider Linear-Time Invariant (LTI) models, although the same ideas are employed when considering more complicated dynamics. In particular, a model of a CPS that uses LTI stochastic difference equations is of the form
\begin{equation}
\left\{
\begin{array}{ll}
{x}(t_{k+1}) = Ax(t_k) + B u(t_k) + v(t_k),  \label{1} \\[1mm]
\hspace{3.8mm} y(t_k) = Cx(t_k) + \eta(t_k),
\end{array}
\right.
\end{equation}
with $k\in \Nat$, sampling time-instant $t_k$, physical state of the system $x \in \Real^n$ (i.e., an $n$-dimentional vector of physical variables associated with the dynamics of the CPS), sensor measurements $y := (y_1,\ldots,y_m)^T \in \Real^m$ , control signals $u := (u_1,\ldots,u_m)^T \in \Real^l$, real-valued matrices $A$, $B$, and $C$ of appropriate dimensions, and i.i.d. multivariate zero-mean Gaussian noises $v \in \Real^n$ and $\eta \in \Real^m$ with covariance matrices $R_{1} \in \Real^{n \times n}$ and $R_2 \in \Real^{m \times m}$, respectively. The matrix $A$ describes the \emph{dynamic evolution of the physical process}, $B$ is used to model the effect of \emph{actuators} on the system dynamics, and the matrix $C$ models the part of the state $x_k$ available from \emph{sensor measurements}.

\subsection{Attacks on Models}
\label{sec:AttacksOnModels}
At the time-instants $t_k,k \in \Nat$, the output of the process $y(t_k)$ is sampled and transmitted over a communication network. The received output  is used to compute control actions $u(t_k)$ which are sent back to the physical process. The complete control-loop is assumed to be performed instantaneously, i.e., the sampling, transmission, and arrival time-instants are supposed to be equal.
In between transmission and reception of sensor data and control commands, an attacker may replace the signals coming from the sensors to the controller and from the controller to the actuators, acting as a \emph{Man-in-the-Middle}. Thus, after each transmission and reception, the attacked output $\bar{y}$ and attacked input $\bar{u}$ take the form
\begin{equation}
\left\{
\begin{array}{ll}
\bar{y}(t_k) := y(t_k) + \delta^y(t_k) 
, \label{3}\\[1mm]
\bar{u}(t_k) := u(t_k) + \delta^u(t_k),
\end{array}\right.
\end{equation}
where $\delta^y(t_k) \in \Real^m$ and $\delta^u(t_k) \in \Real^l$ denote \emph{additive sensor and actuator attacks}, respectively.

Henceforth, we denote $x(t_k)$ by $x_k$, $ \bar{u}(t_k)$ by $\bar{u}_k$, $v(t_k)$ by $v_k$, $\bar{y}(t_k)$ by $\bar{y}_k$, $\eta(t_k)$ by $\eta_k$, $\delta^y(t_k)$ by $\delta_k^y$, and $\delta^u(t_k)$ by $\delta^u_k$. Then, a system under attack is modelled by
\begin{equation}
\left\{
\begin{array}{ll}
{x}_{k+1} = A x_k + B (u_k + \delta^u_k) + v_k,\label{17} \\
\text{ \ \ }\hspace{.35mm}\bar{y}_k = C x_k + \eta_k + \delta^y_k.
\end{array}
\right.
\end{equation}

\subsection{State Estimation}

The vast majority of work on attack detection uses fault detection techniques \cite{Marios_Poly,Patton_1} to identify attacks. The main idea behind fault detection theory is the use of an estimator to forecast the evolution of the system. If the difference between what it is measured and the estimation is larger than expected, then there may be a fault/attack on the system.

For LTI CPSs of the form \eqref{1}, to estimate the state of the physical process, people from the systems and control community use mainly two types of estimators: the \emph{Kalman filter} \cite{Astrom} and the \emph{Luenberger observer} \cite{1098323}. Here, we consider Luenberger observers of the form
\begin{equation}\label{19}
\left\{ \begin{array}{ll}
\hat{x}_{k+1} = A \hat{x}_k + B\bar{u}_k + L \big( \bar{y}_k - C\hat{x}_k   \big),\\
\hspace{3.0mm} \hat{y}_k = C \hat{x}_k,
\end{array} \right.
\end{equation}
with estimated state $\hat{x}_k \in \Real^n$ and \emph{observer gain matrix} $L \in \Real^{n \times m}$ which must be designed to ensure that the estimated state $\hat{x}_k$ converges to the actual state $x_k$ asymptotically, i.e., $L$ is selected such that $\lim_{k \rightarrow \infty} (x_k - \hat{x}_k) = 0$. Define the \emph{estimation error} $e_k:= x_k - \hat{x}_k$ and the \emph{residual sequence} $r_k \in \Real^m$ as
\begin{align}
r_k := \bar{y}_k - C\hat{x}_k = Ce_k + \eta_k + \delta_k^y. \label{25}
\end{align}
Given the system dynamics (\ref{17}) and the observer (\ref{19}), the estimation error dynamics and the residual sequence are governed by the difference equation:
\begin{equation}
\left\{
\begin{array}{ll}
e_{k+1} = \big( A - LC \big) e_k  + v_k - L\eta_k  - L \delta_k^y + B\delta_k^u,  \label{26} \\[.5mm]
\text{ \ \ }\hspace{.5mm} r_k = Ce_k + \eta_k + \delta_k^y.
\end{array}
\right.
\end{equation}
At this point, having introduced the residual dynamics \eqref{26}, the main idea behind residual-based (model-based) attack/fault detection procedures is to characterize the `behavior' of the residual sequence \emph{a priori} in the attack-free case. Then, we can test (in real-time) whether the actual behavior of the CPS matches the attack-free one. If it is not the case, alarms are raised indicated a possible fault/attack on the system. More precisely, the \emph{asymptotic density distribution} \cite{Ross} of the residual sequence is obtained given the system and observer matrices $(A,B,C,R_1,R_1,L)$; then, statistical change detection procedures (e.g., Cumulative Sum (CUSUM) \cite{Gustafsson,Page}, Genera\-lized Likelihood Ratio (GLR) testing \cite{Basseville}, Compound Scalar Testing (CST) \cite{Gertler}, etc.) are employed to test whether the residual $r_k$, at every time step $k$, belongs to this distribution. If it is not the case, alarms are raised. For transparency, we do not go into further details about these ideas. However, for the interested reader, we include some extra material about these techniques in the appendix.


\section{Quantitative Integrity Analysis on CPS}
Traditional IT security mechanisms such as authentication and message integrity are insufficient for CPS security. However, as stated in \cite{CPSAttacksAgainstPCS}, the major distinction of control systems with respect to other IT systems is the interaction with the physical world. In \cite{CPSSecVinyl}, the authors state that CPSs security is {specifically} concerned with attacks that cause physical impact; since IT security mechanisms do not usually account for the physical part of the system, they are thus ineffective against attacks that either target or exploit the physical components of CPSs. In this section, we present an Information Flow (IF) inspired integrity analysis for CPSs which accounts for the physical process.


\subsection{A Discrete Model for CPSs}
Figure \ref{fig:IFCPS} presents the \emph{supervisor model} \cite{doi:10.1137/0325013}, which serves as the starting point to define an IF framework for CPSs.
\begin{figure}
\centering
\includegraphics[width=0.45\textwidth]{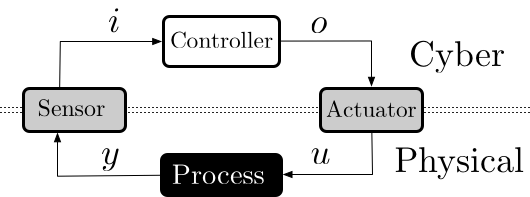}
\caption{Supervisor model}
\label{fig:IFCPS}
\end{figure}

The supervisor model consists of four systems: a discrete cyber system, \emph{i.e.}, the \emph{controller}, a continuous-time physical system \emph{i.e.}, the \emph{process}, and two cyber-and-physical subsystems \emph{i.e.}, the \emph{sensor} and the \emph{actuator}. The controller, whose state is $q$, receives a digital input $i$ and produces a digital control signal $o$ which is then transformed into an analog control signal $u$ by the actuator; in turn, the process, whose state is $x$, reacts to $u$ and yields an analog observation $y$, which is transformed into the digital input $i'$ by the sensor, and the cycle repeats. Ultimately, our objective is to protect the integrity of the physical variable observed by $y$; i.e., (a part of) the state $x$ of the process. 

To formally reason about this model, we define a notion of  state as follows. 
\begin{definition}[States of a CPS]
\label{def:States}
Let $i$ be the vector of digital inputs, $o$ be the vector of digital outputs, $u$ be the vector of analog control signals, $q$ be the current state of the controller, $x$ be the current state of the process, and $y$ be the vector of analog observations. A \emph{state of the CPS} is a tuple
\begin{align}
\sigma=((i,o,u),(q,x),y)
\end{align}
where $(i,o,u)$ is the \emph{controllable} part of $\sigma$, $(q,x)$ is the \emph{hidden} part of $\sigma$ and $y$ is the \emph{observable} part of $\sigma$. We assume that the behaviours of the actuator and of the sensor are time-invariant; thus we do not track their state.

Let $I$ be the set of all values that $i$ can take. Similarly, we define $O$, $U$, $Q$, $X$ and $Y$ as the sets of all values that $o$, $u$, $q$ $x$ and $y$ can take, respectively. We define the set of all states $\Sigma=I\times O\times U\times Q \times X \times Y$.
\end{definition}
\begin{example}[A Water Tank]
\label{exampleTank}
Consider the water tank with one input valve and one output valve shown in Figure \ref{fig:Tank}. The tank has a capacity of $L$ litres of water, and a water level sensor measures the current level of water in the tank. A very simple LTI model of the process is 
\begin{align}
\label{eq:TankDynamics1}
{x}_{k+1}&= {x}_{k} + 5{u}^{in}_{k}-3{u}^{out}_{k},\\
\label{eq:TankDynamics2}
{y}_{k}&= {x}_{k},
\end{align}
subject to $0\leq {x}_{k}\leq L$. 
If the ${in}$ valve is open, then $5$ litres of water flow into the tank, and if the ${out}$ valve is open, then $3$ litres of water flow out the tank. For simplicity, there is no noise, and the only part of the state that we are interested in is the water level (not temperature or pressure, for example).

The controller of the tank has a single state $*$, and it receives a single input $i$, \emph{i.e.}, the level of the tank in digital format, and issues an output of the form $(o^{in}, o^{out})$ where $o^{in}$ and $o^{out}$ are each either ${open}$ or ${close}$, which respectively represent the actions of opening and closing the ${in}$ and ${out}$ valves. Opening valve $v$ at time $k$ forces $u^{v}_k$ to be $1$ and closing valve $v$ at time $k$ forces $u^{v}_k$ to be $0$. A state for this system at time $k$ could be $\AsTuple{(0,(open,close),(1,0)),(*,3),3}$.

\begin{figure}
\centering
\includegraphics[scale=0.5]{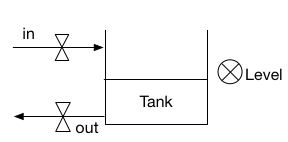}
\caption{Water Tank Model. We denote valves by $\bowtie$ and sensors by $\otimes$.}
\label{fig:Tank}
\end{figure}
\end{example}
We model the behaviour of the sensor, the actuator, the controller and the process as a set of black-box functions, which we compose in order to define the dynamics of the CPS.

{
\begin{definition}[Single-Cycle Semantics]
\label{def:SingleCycleSemantics}
Let $\Sigma=I\times O\times U\times Q \times X \times Y$ be the set of states of the CPS. Now, consider the following functions:
\begin{itemize}
\item $\delta_{IO}\colon Q\times I \rightarrow Q$, the transition function of the controller,
\item $\gamma_{IO}\colon Q\times I \rightarrow O$, the output function of the controller,
\item $\gamma_{AD}\colon Y \rightarrow I$, the output function of the sensor,
\item $\gamma_{DA}\colon O \rightarrow U$, the output function of the actuator,
\item $\gamma_{\Phi}\colon X\times U \rightarrow Y$, the output function of the process,
\item $\delta_{\Phi}\colon X \times U \rightarrow X$, the transition function of the process.
\end{itemize}
We define the single-cycle semantics function $\TheSystem\colon \Sigma \rightarrow \Sigma$ that describes the evolution of the state $((i,o,u),(q,x),y)$ of the CPS as follows
\begin{align}
\TheBehaviourOf{(i,o,u),(q,x),y)}&=\left((i', o', u'), (q',x'),y' \right)
\end{align}
where 
\begin{gather}
i'=\gamma_{AD}(y), \quad o'=\gamma_{IO}(q,i), \quad u'= \gamma_{DA}(o)\\
q'=\delta_{IO}(q,i), \quad x'=\delta_{\Phi}(x,u)\\
y'=\gamma_{\phi}(x,u).
\end{gather}
The evolution over $k+1$ cycles is defined by iterating $\TheSystem$, \emph{i.e.},
\begin{align}\label{eq:kCycles}
\TheSystem^{k+1}&= \TheSystem \circ \TheSystem^{k}.
\end{align}
\end{definition}

}


\subsection{Integrity and Attacker Model}
We have yet to consider the presence of an attacker. There are three non-excluding cases: 1) if an attacker controls a sensor, then the process ${IO}$ receives a corrupted version of $i$, 2) if an attacker physically controls an actuator, then ${\Phi}$ receives a corrupted version of $u$, and 3) if an attacker acts as a Man-in-the-Middle between the controller and the actuator, then the actuator receives a corrupted version of $o$. To formally reason on the impact of such attacks, we characterise integrity and an attacker model.  

Attackers that control all sensor and/or control signals, or that can change the logic of controller can drive the process to any state they desire
; these attackers are too powerful to defend against. However, it may happen that an attacker that only controls a subset of sensors and actuators has enough power to drive the process to any state they want, and defending against these attackers requires the implementation of orthogonal controls and/or reimplementation of the controller's logic, and is thus desirable to detect these vulnerabilities at the design phase of the CPS.

Clarkson and Schneider state in \cite{QuantitativeIntegrity} that they know of no widely accepted definition of integrity, but they remark that an ``informal definition seems to be the `prevention unauthorized modification of information'." Clarkson and Schneider use two formal notions to characterise and quantify \emph{corruption}, \emph{i.e.}, the damage to integrity. One notion is \emph{contamination}, which is a generalisation of taint analysis that tracks information flows from untrusted inputs to outputs that are supposed to be trusted. The other notion is \emph{suppression}, which occurs when implementation outputs omit information about correct outputs (with respect to the specification). Contamination is the dual to information leakage under the Biba duality; however, there is no apparent Biba dual to suppression \cite{BibaIntegrity}. In this work, the intuition behind corruption is closer to that of contamination, and we assume that the system presents no suppression; \emph{i.e.}, the CPS and its components are correctly designed.


Goguen and Meseguer introduced the concept of \emph{noninterference} \cite{Noninterference}, which provides a formal foundation for the specification and analysis of security policies and the mechanisms to enforce them. In the case of integrity, noninterference for programs means that a variation of public inputs does not cause a variation of critical values. Noninterference is traditionally used to characterise unwanted flows from a ``high/secret'' security level $\High$ to a ``low/public'' security level $\Low$, as they violate the confidentiality of data; however, given the Biba duality, in the case of integrity we say that the inputs of the attacker belong to $\Low$, that critical values belong to $\High$, and that our objective is to prevent unwanted flows from $\Low$ to $\High$, as they violate the integrity of data. Noninterference-based confidentiality notions have been formally defined over state-based models (see \cite{Rushby92}), trace-based models (see \cite{Hyperproperties}), and semantics-based models (see \cite{LanguageBasedInformationFlowSecurity}) to name a few. In the following, we define a notion of (noninterference-based) integrity for CPS using a semantics-based model.

\begin{definition}[Attack]
Let $\Sigma=I\times O\times U\times Q \times X \times Y$ be the set of  states of the CPS. An \emph{attack} is a function $f\colon I\times O \times U \rightarrow I\times O \times U$ that changes the \emph{controllable components} of a given state $\sigma$ of the CPS. The extension of an attack on a single state can be naturally extended to an attack over a trace $\tau = (\sigma_0,\sigma_1, ...) \in \Sigma^*$, where the attacker can apply arbitrary functions $f_j$ on each 
state $\sigma_i$ of the trace. 

\end{definition}
This definition of attacks is inspired by Equation \eqref{3},  where attackers change the value of control signals and measurements. We remark that these attack functions do not directly change the current state of the controller nor the state of the process; instead, they attempt to model the interference caused by the attacker at either the cyber or physical level. For instance, if the CPS is in a state $((i,o,u),(q,x),y)$, then an attack on a sensor yields a state $((i',o,u),(q,x),y)$, and an attack on the network layer that communicates the controller with an actuator yields a state $((i,o',u),(q,x),y)$. This interference must propagate in order to change the state of the controller or the state of the process, which is why we believe performing an information flow analysis is both sensible and adequate in this case. 
\begin{example}[Attacking the Water Tank]
Consider the water tank from Example \ref{exampleTank}. An example attack on the water level sensor is characterised by the function that maps the controllable tuple $(i,o,u)$ to the tuple $(L,o,u)$; in this case, the attacker is going to fool the controller into thinking that the water tank is full. 
\end{example}

\begin{definition}[Attacker]
\label{def:Attacker}
Every attacker $\alpha$ controls a set of controllable components, \emph{i.e.}, a set of digital inputs, digital outputs, and/or analog control signals. We denote the set of the components that $\alpha$ controls by $C_\alpha$.
 The attacker $\alpha$ can perform an attack sequence $f_1, f_2, \ldots f_{k}$ over $k$ cycles that changes the semantics of Equation \eqref{eq:kCycles} to
\begin{align}\label{eq:newkCycle}
\TheSystem^{k+1}&= \TheSystem \circ f_{k}\circ\TheSystem^{k}.
\end{align}
%
\end{definition}
\begin{example}[An Attacker of the Water Tank]
Consider the water tank from Example \ref{exampleTank}. An example attacker $\alpha$ that only controls $u^{in}$ models an attacker that can physically open or close the $in$ valve at will, and the only means for this attacker to attack the system is through the physical manipulation of such valve.
\end{example}

The objective of an attacker is to control one or more of the \emph{process variables} of the CPS; \emph{e.g.}, an attacker wants to control how much water is in a tank so that she can overflow it. These physical variables are represented in our CPS model by the output of the process, $y$, which is a vector of their analogue observations. Ultimately, our notion of security will depend on the vector $y$. We now formally define security based on classical notions of noninterference.
\begin{definition}[$k$-Process Integrity for CPSs]
\label{def:NICPS}
Let $\alpha$ be an attacker.
We label each component controlled by $\alpha$, \emph{i.e.}, $c\in C_\alpha$ as $\Low$, and we label all other components of the state as $\High$. We use the expression $\Lvl(c)$ to denote the security level of the component $c$.

Now, let $\sigma$ and $\sigma'$ be  states of the CPS. We say that $\sigma$ and $\sigma'$ are \emph{$\High$-equivalent} iff they are equal in all $\High$-labeled components; formally, we say that  $\sigma =_{\High} \sigma'$  iff
\begin{align}
\forall c \in C \colon \Lvl(c)=\High \Rightarrow \sigma(c)=\sigma'(c).
\end{align}
We say that $\sigma$ and $\sigma'$ are \emph{process-equivalent}, denoted ${\sigma}\approx_{\Phi} {\sigma'}$ iff their process variables are equal; more precisely, if
\begin{align}{\sigma}=((i,o,u),(q,x),y),\text{ and}\\{\sigma'}=((i',o',u'),(q',x'),y'),
\end{align}
we say that ${\sigma}\approx_{\Phi}{\sigma'}$ iff $y = y'$.

Finally, for $k\in \Nat$ with $k>0$, 
we say that the CPS satisfies \emph{$k$-
process integrity against the attacker $\alpha$} if and only if, for all states $\sigma
$ and $\sigma'
$
, 
we have
\begin{align}
\label{eq:NoninterferenceCPS}
\sigma =_{\High} \sigma' \Rightarrow \TheBehaviourOf{\sigma}^k \approx_{\Phi} \TheBehaviourOf{\sigma'}^k.
\end{align}
\end{definition}
The previous definition is qualitative (either the system satisfies process integrity or it does not). However, it does not shed light on the level of damage that this integrity violation may imply. As with many security notions based on noninterference, an attacker might slightly deviate the behaviour of the process without any serious repercussions; thus, we are interested in a quantitative measure of security in order to precisely estimate the level of control that the attacker has over the process.
\begin{definition}[$k$-Controllability of an Attacker]
\label{def:QuantitativeIntegrity}
Let $\alpha$ be an attacker, let $S$ be the set of states of the CPS, and let $\sigma_0$ be the initial state of the CPS. We inductively define the set of states $\Sigma^{k}_\alpha$, for $k \in \Nat$, by
\begin{align}
\Sigma^0_\alpha &= \set{\sigma_0}\\
\Sigma^{k+1}_\alpha &= \Set{\TheBehaviourOf{\sigma}|\sigma \in \Sigma,  \sigma ' \in \Sigma^k_\alpha \colon \sigma =_\High \sigma'}.
\end{align}
In other words, for each element $\sigma' \in \Sigma^k_\alpha$, we find every $\High$-equivalent state $\sigma$, and we include the resulting state $\TheBehaviourOf{\sigma}$ in $\Sigma^{k+1}_\alpha$. Thus, the set $\Sigma^{k+1}_\alpha$ contains all reachable states after the $k^{th}$ cycle accounting for all attacks that $\alpha$ could use against the CPS during these $k$ cycles (according to the interaction between the attacker and the system shown in Equation \ref{eq:newkCycle}).

Now, consider the function $\pi_Y\colon \Sigma \rightarrow Y$ that projects the state $((i,o,u),(q,x),y)$ to $y$; we use the set 
\begin{align}
\pi_Y(\Sigma^{k})=\Set{\pi_Y(\sigma)|\sigma \in \Sigma^{k}}
\end{align}
to determine how many different values for $y$ the attacker $\alpha$ can force the system to have after $k$ cycles. The \emph{$k$-controllability of $\alpha$} defined by the number of elements in $\pi_Y(\Sigma^{k})$, \emph{i.e.}, $|\pi_Y(\Sigma^{k})|$, quantifies the degree of control that the attacker $\alpha$ has over the process after a period of $k$ cycles. 
\end{definition}

Note that this definition, being quantitative, allows us to compare the impact of different attackers on a CPS over well-defined control logics. Note also that however, whether this level of controllability is enough to drive the CPS to a \emph{critical state} requires an additional analysis as we will discuss in the following section. Nevertheless, as in the case of traditional \emph{leakage}, we observe that higher levels of controllability (that can be 
thought of as a dual of leakage, or \emph{corruption}), are in general more 
dangerous, since they imply a higher degree of separation from the intended  ``legitimate'' state.

We remark that we can show that a system does not satisfy process integrity against an attacker $\alpha$ if, for some $k$, the $k$-controllability of $\alpha$ is greater than $1$ , as shown in the following Corollary.
\begin{corollary}[From Controllability to Integrity]
Let $\alpha$ be an attacker. For $k \in \Nat$, if $|\pi_Y(\Sigma_\alpha^{k})|>1$, then the system does not satisfy $k$-process integrity against $\alpha$.
\end{corollary} 
\begin{proof}
If $|\pi_Y(\Sigma_\alpha^{k})|>1$, then there exist $\sigma\in \Sigma$ and $\sigma ' \in \Sigma_\alpha^{k-1}$ with $\sigma =_{\High} \sigma'$ such that $\TheBehaviourOf{\sigma} \not \approx_\Phi \TheBehaviourOf{\sigma'}$. According to Equation \eqref{eq:NoninterferenceCPS}, this would imply that the CPS does not satisfy $1$-process integrity against $\alpha$.
%
%
\end{proof}
We are ultimately interested in using the measure $|\pi_Y(\Sigma_\alpha^{k})|$ to reduce the controllability that the attacker has over the process by changing the logic of the controller. We will discuss this in Section \ref{sec:Redesign} by means of a simple but not trivial case study. 


\subsection{Composition of CPSs}
Cyber-Physical Systems are often distributed systems that compose several control units supervising certain aspects of the process and communicating with 
each other. Thus, it might be possible for an attacker to attack a physical component in a module of a CPS by interfering with a digital component of a different module. Intuitively, this requires the existence of channels that enable the propagation of the interference caused by the attacker from module to module. These channels can be either digital or physical: when two PLCs communicate, their message passing opens a digital channel, but they might also supervise or control a physical ``channel'' (such as the same water pipe at different points). 

Consider the composition of two modules shown in Figure \ref{fig:Composition}; we assume that an attacker exists somewhere on the left module, and her objective is to affect the process variables of the right one. For that purpose, she will try to propagate her interference through the channels $i^*$ (digital) and $u^*$ (analogue). 
\begin{figure}[h]
\centering
\includegraphics[scale=0.25]{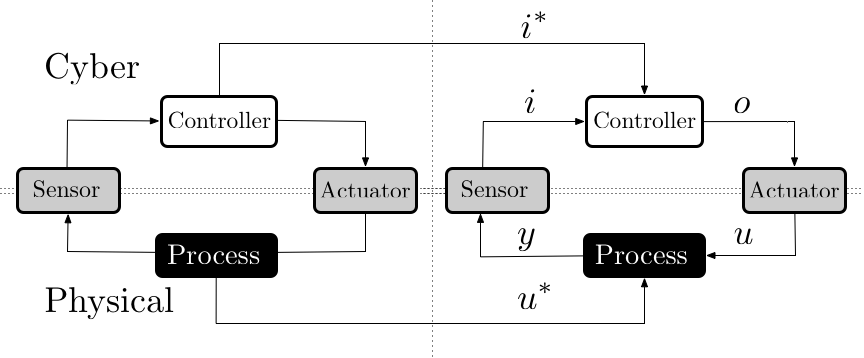}
\caption{Composition framework}
\label{fig:Composition}
\end{figure}

To analyse the module on the right of Figure \ref{fig:Composition}, we can extend its  state from $((i,o,u),(q,x),y)$ to $(((i,i^*),o,(u,u^*))(q,x),y)$; \emph{i.e.}, we extend the vectors $i$ and $u$ to now consider the new channels. Consequently, we need not compose the two modules in their entirety to perform an analysis, but only provide additional semantics for $i^*$ and $u^*$ (which could be given by state machines); this opens the door to a modular analysis of networks of CPSs, which can provide compositional results under some assumption/commitment restrictions. 

\section{Case Study: a Water Distribution CPS}
\label{sec:CaseStudy}
We now study a simple but illustrative example on how to provide useful design insights by means of our formal approach. 

\begin{figure}
\centering
\includegraphics[scale=0.3]{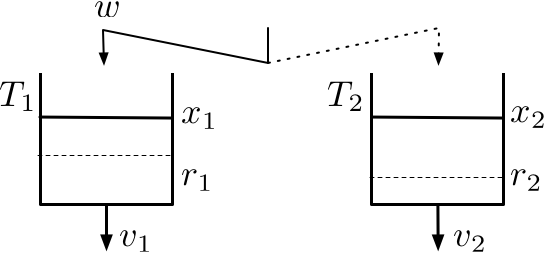}
\caption{Case study water distribution system.}
\label{fig:CaseStudy}
\end{figure}


The following case study is based on an example taken from \cite{ModelingAndControlOfHybridDynamicalSystems}. Consider the system shown in Figure \ref{fig:CaseStudy} which consists of two water tanks, $T_1$ and $T_2$, both with maximum capacity $L$. Each tank has a constant outflow of $v_1$ and $v_2$, respectively. Water is added to the system at a constant rate $w$ through a hose which, at any point in time is dedicated to either $T_1$ or $T_2$. We assume that the hose can switch between the tanks instantaneously, and we also assume that $w\geq v_1+v_2$, so it is also possible to close the hose to prevent overflowing either tank. The objective of the controller is to constantly keep the water volume $x_t$ above some time-invariant critical threshold $r_t$, for $t=1,2$. We assume that the water volumes are both initially above their respective thresholds. For this purpose, the controller is designed to switch the inflow to $T_1$ whenever $x_1< r_1$ and to $T_2$ whenever $x_2 < r_2$. 

The process has three modes: mode $q_1$ where $T_1$ is filling, mode $q_2$ where $T_2$ is filling, and mode $q_0$ where the hose is closed. The process evolves as follows when in mode $q_1$
\begin{align}
\begin{cases}
x_1(k+1)=x_1(k)+w-v_1,\\
x_2(k+1)=x_2(k)-v_2,
\end{cases}
\end{align}
 as follows when in mode $q_2$
\begin{align}
\begin{cases}
x_1(k+1)=x_1(k)-v_1,\\
x_2(k+1)=x_2(k)+w-v_2,
\end{cases}
\end{align}
and as follows when in mode $q_0$
\begin{align}
\begin{cases}
x_1(k+1)=x_1(k)-v_1,\\
x_2(k+1)=x_2(k)-v_2.
\end{cases}
\end{align}
At time $k$, the process outputs the vector of analog measurements $(y_1(k), y_2(k))$, defined by
\begin{align}
(y_1(k), y_2(k))=(x_1(k),x_2(k));
\end{align}
in this case, we see that the state is fully observable (though this might not always be the case). The sensor, at time $k$, receives the vector of measurements $(y_1(k), y_2(k))$ and outputs the vector of digital inputs $(i_1(k), i_2(k))$, defined by
\begin{align}
(i_1(k), i_2(k))=(y_1(k),y_2(k));
\end{align}
Although it may seem redundant, we explicitly define the semantics of the sensor for the sake of clarity. In other scenarios, the sensor may have more complicated semantics; \emph{e.g.}, it could encrypt or sign the data before sending it to the controller. 

The controller, which has a unique state $*$, receives the vector of inputs $(i_1, i_2)$ and outputs a command $o$, defined by
\begin{align}
\label{eq:controller}
\mathtt{if}\  i_1 < r_1\ \mathtt{then}\ o:= q_1,\\
\mathtt{else if}\  i_2 < r_2\ \mathtt{then}\ o:= q_2,\\
\mathtt{else }\  o:= q_0.
\end{align}
Then, at time $k$, the actuator translates the command $o(k)$ into the analog control signal $u(k)$ defined by 
\begin{align}
u(k)=o(k).
\end{align}
The physical process works as a state machine that reacts to the control signal $u(k)$ and \emph{switches} on it to decide how to update the state; \emph{i.e.}, if $u(k)=q_j$, then update the state using the equations of mode $q_j$, for $j=0,1,2$. This can be clearly represented as a {hybrid automaton} (see \cite[Figure 1.17]{ModelingAndControlOfHybridDynamicalSystems}). We now define some attackers and analyse their impact over this CPS.


\subsection{Attackers}
We want to determine whether the an attacker can violate the integrity of the process (according to Definition \ref{def:NICPS}), and if so, quantify the potential impact. This impact will be a natural number, that will allow us to compare various attackers and assess different controllers in terms of security. In order to complement this quantitative analysis, we will also consider consider four classes of critical states: for $t=1,2$, when tank $T_t$ is empty (classes $E_1$ and $E_2$), and when tank $i$ is full (classes $F_1$ and $F_2$). 

According to Definition \ref{def:States}, states of this CPS have the structure $((i_1,i_2,o,u),(*,(x_1,x_2)),(y_1,y_2))$. For $t=1,2$, we say that the  state $\sigma$ is a critical state of class $E_t$ at time $k$ (i.e., a state where tank $i$ is empty) iff $y_t(k)=0$; similarly, we say that $\sigma$ is a critical state of class $F_t$ at time $k$ iff $y_t(k)=L$. The system does not reach any critical state under normal circumstances since the controller is designed to stabilise the values of $y_1$ and $y_2$ around the thresholds $r_1$ and $r_2$, respectively. Thus, for $t=1,2$ the normal range of values for $y_t$ is $[r_t^-, r_t^+]$, with $r_t^-\leq r_t \leq r_t^+$, and we will assume for the initial value $y_t(0)$ that $r_t \leq y_t(0)\leq r_t^+$. Note that in general our assumption is that the controller under analysis has been already 
verified to comply with their functional and safety requirements under normal circumstances (no attackers present).

Let us consider three different attackers:\\
 \textbf{Attacker $\mathbf{\alpha_1}$}, who controls the digital control signal $o(k)$, for all $k$.\\
\textbf{Attacker  $\mathbf{\alpha_2}$}, who controls the digital measurement $i_2(k)$, for all $k$.\\
\textbf{Attacker  $\mathbf{\alpha_3}$}, who controls the digital  measurement $i_1(k)$, for all $k$.

Note that since we want to protect the physical state of the system (level 
of the tanks), we do not consider attackers that can directly tamper with this state, since they will trivially achieve their goal. 

\subsection{Quantification and analysis}

Using the formal semantics and security properties described above, it is now possible to carry out a mathematical analysis on the impact that the three attackers considered can have on the system. Table \ref{tab:Attackers} summarises the results of this analysis. Attacker $\alpha_1$ is very powerful since she can drive the system to any class of critical state she wants. Attackers $\alpha_2$ and $\alpha_3$ are very similar to each other in terms of capabilities, but $\alpha_3$ is more powerful than $\alpha_2$, since $\alpha_3$ can take the system to three out of the four classes of critical states, while $\alpha_2$ can only take the system to two classes of critical states.  In the following we sketch manual proofs of these results, and we start to build arguments for the discussion we carry out in Section \ref{sec:Discussion}.

{\small
\begin{table}[!h]
\centering
\begin{tabular}{|c|c|c|c|c|c|c|c|}
\hline
 \multirow{2}{*}{Attacker} &  \multirow{2}{*}{Controls} &  \multicolumn{2}{c|}{Quantification} & \multicolumn{4}{c|}{Vulnerable?}\\\cline{3-8}
 &	&	$y_1$& $y_2$& $E_1$&	$E_2$	&	$F_1$&   $F_2$\\
\hline\hline
$\alpha_1$	&	${o}$ &$[0,L]$&$[0,L]$& $\cmark$&$\cmark$&$\cmark$&$\cmark$ \\
\hline
$\alpha_2$	&	${i_2}$ &$[r_1^-,r_1^+]$ &$[0,L]$&$\xmark$&$\cmark$&$\xmark$&$\cmark$ \\
\hline
$\alpha_3$	&	${i_1}$ &$[0,L]$&$[0,r_2^+]$ & $\cmark$&$\cmark$&$\cmark$&$\xmark$ \\
\hline
\end{tabular}
\vspace{1em}
\caption{Quantification of $k$-controllability of the attackers for a $k$ that is large enough. We quantify by using ranges to over-approximate the sets $\Sigma^k_\alpha$ for the different attackers.} 
\label{tab:Attackers}
\vspace{-10pt}
\end{table}
}

\noindent\emph{Proof outline:} We assume that the initial state is
\begin{align}\sigma(0)=((r_1^+,r_2^+),q_0,q_0),(*,(r_1^+,r_2^+)),(r_1^+,r_2^+)),
\end{align} \emph{i.e.}, a state where the water level of the tank $T_t$ is $r^+_t$, for $t=1,2$, and the hose is closed.\\
\textbf{Attacker $\alpha_1$} has control over the digital output $o$, and she can choose to replace $o(k)$ with an element of $\set{q_1, q_2, q_0}$ for all $k$.  In one cycle, we have 
\begin{align*}
\Sigma^1_{\alpha_1}=&\left\{\TheBehaviourOf{((r_1^+,r_2^+),\underline{q_0},q_0),(*,(r_1^+,r_2^+)),(r_1^+,r_2^+))}, \right.\\
&\left. \TheBehaviourOf{((r_1^+,r_2^+),\underline{q_1},q_0),(*,(r_1^+,r_2^+)),(r_1^+,r_2^+))}, \right.\\
&\left. \TheBehaviourOf{((r_1^+,r_2^+),\underline{q_2},q_0),(*,(r_1^+,r_2^+)),(r_1^+,r_2^+))}\right\}\\
=&\left\{{((r_1^+,r_2^+),q_0,\underline{q_0}),(*,(r_1^+-v_1,r_2^+-v_2)),(r_1^+-v_1,r_2^+-v_1))}, \right.\\
&\left.{((r_1^+,r_2^+),q_0,\underline{q_1}),(*,(r_1^+-v_1,r_2^+-v_2)),(r_1^+-v_1,r_2^+-v_1))}, \right.\\
&\left.{((r_1^+,r_2^+),q_0,\underline{q_2}),(*,(r_1^+-v_1,r_2^+-v_2)),(r_1^+-v_1,r_2^+-v_1))} \right\}.
\end{align*}
(We underline the differences.)
Since $\pi_{y_1}(\Sigma^1_{\alpha_1})=\set{r_1^+-v_1}$ and $\pi_{y_2}(\Sigma^1_{\alpha_1})=\set{r_2^+-v_2}$, we see that the attack issued by $\alpha_1$ has not had enough time to propagate to the process. However, after considering one more cycle, we see that 
\begin{align}
\pi_{y_1}(\Sigma^2_{\alpha_1})=\set{r_1^+-2v_1, r_1^+-2v_1+w},\\
\pi_{y_2}(\Sigma^2_{\alpha_1})=\set{r_2^+-2v_2, r_2^+-2v_2+w}.
\end{align}
We see two different values for $y_1$ because the control signals $q_0$ and $q_2$ prevent the tank $T_1$ from filling (thus making the level of water $r_1^+-2v_1$) while the control signal $q_1$ causes the hose to fill $T_1$ (making the level of water $r_1^+-v_1+w$). As we extend our analysis over several cycles (a large enough $k$), we see that we obtain linear combinations of $v_t$ and $w$, for $t=1,2$, as follows
\begin{align}
\pi_{y_1}(\Sigma^k_{\alpha_1})=\set{y | y=r_1^+-\beta v_1+\epsilon w, 0\leq y \leq L, \beta, \epsilon \in \Nat},\\
\pi_{y_2}(\Sigma^k_{\alpha_1})=\set{y | y=r_2^+-\beta v_2+\epsilon w, 0\leq y \leq L, \beta, \epsilon \in \Nat}.
\end{align}
For $t=1,2$, the set $\pi_{y_t}(\Sigma^k_{\alpha_1})$ includes $0$ and $L$, because there will be linear combinations $r_t^+-\beta v_t+\epsilon w$ that would be less than $0$ and greater than $L$; however, the level of water in the tank $T_t$ cannot go below 0 and it cannot go above $L$. For simplicity of exposition, we over-estimate these sets using the interval $[0,L]$ (see Table \ref{tab:Attackers}).

\textbf{Attacker $\alpha_2$} has control over the digital input $i_2$, and she could choose to replace $i_2(k)$ with an element from the range $[0,L]$ for all $k$. However, we see that many choices for $i_2$ trigger similar behaviours inside the controller. Since there is only one branching depending on whether $i_2 < r_2$, we take $0$ and $L$ as the two representative values for $i_2$ that $\alpha_2$ could use to trigger different behaviours. Thus in one cycle, we have 
\begin{align*}
\Sigma^1_{\alpha_2}=&\left\{\TheBehaviourOf{((r_1^+,\underline{r_2^+}),q_0,q_0),(*,(r_1^+,r_2^+)),(r_1^+,r_2^+))}, \right.\\
&\left. \TheBehaviourOf{((r_1^+,\underline{0}),q_0,q_0),(*,(r_1^+,r_2^+)),(r_1^+,r_2^+))}, \right.\\
&\left. \TheBehaviourOf{((r_1^+,\underline{L}),q_0,q_0),(*,(r_1^+,r_2^+)),(r_1^+,r_2^+))}\right\}\\
=&\left\{{((r_1^+,r_2^+),\underline{q_0},q_0),(*,(r_1^+-v_1,r_2^+-v_2)),(r_1^+-v_1,r_2^+-v_1))}, \right.\\
&\left.{((r_1^+,r_2^+),\underline{q_2},q_0),(*,(r_1^+-v_1,r_2^+-v_2)),(r_1^+-v_1,r_2^+-v_1))}, \right.\\
&\left.{((r_1^+,r_2^+),\underline{q_0},q_0),(*,(r_1^+-v_1,r_2^+-v_2)),(r_1^+-v_1,r_2^+-v_1))} \right\}.
\end{align*}
At $k=3$, there is finally an impact on the value of $y_2$, and we have
\begin{align}
\pi_{y_1}(\Sigma^3_{\alpha_2})&=\set{r_1^+-3v_1},\\
\pi_{y_2}(\Sigma^3_{\alpha_2})&=\set{r_2^+-3v_2, r_2^+-3v_2+w}.
\end{align}
In this case, the attacker $\alpha_2$ has the same degree of control over $y_2$ that $\alpha_1$ had: $\alpha_2$ can cause the water level of $T_2$ to rise or decrease at will (although with slightly more delay). However, there is no flow of information from $i_2$ to $y_1$, so we the values for $y_1$ correspond to the values that we we see during the normal execution of the system over $k$ cycles, \emph{i.e.}, values between $r_1^-$ and $r_1^+$.
{
\\
\textbf{Attacker $\alpha_3$} has control over the digital input $i_1$, and she could choose to replace $i_1(k)$ with an element from the range $[0,L]$ for all $k$. The influence of $\alpha_3$ over $y_1$ is clear, since it is dual to the influence that the attacker $\alpha_2$ had over $y_2$, but we also see an indirect flow from $i_1$ to $y_2$ due to the way the controller is programmed. More precisely, whenever $i_1< r_1$, the controller never outputs $q_2$, so it is possible for the attacker to continuously inhibit the mode $q_2$ by, for example, making $i_1=0$. Consequently, the value $r_2^+-kv_2$ would be an element of $\pi_{y_2}(\Sigma^k_{\alpha_3})$. Since the value $r_2^+-kv_2$ cannot be less than 0, whenever $r_2^+\leq kv_2$, we conclude that $0$ is an element of $\pi_{y_2}(\Sigma^k_{\alpha_3})$ and that $\alpha_3$ can empty $T_2$. However, no value of $i_1$ forces the controller to output $q_2$, so $\alpha_3$ cannot arbitrarily increase the value of $y_2$ beyond its operational upper bound $r_2^+$.
}
\subsection{Redesigning the Controller}
\label{sec:Redesign}
In this section, we show how we can reduce the controllability of the attacker $\alpha_3$ by adding a fairness mechanism to the controller. Instead of a single mode $*$, we now say that the controller has two. Let $m$ be a variable that denotes the current mode of the controller, whose value is either $m_1$ or $m_2$. The new controller receives the vector of inputs $(i_1, i_2)$ and outputs a command $o$ depending on its current mode. If the mode is $m_1$ then
\begin{align}
\label{eq:controllernew1}
\mathtt{if}\  i_1 < r_1 \land  i_2< r_2 \ \mathtt{then}\ o:=q_1, m:=m_2,\\
\mathtt{else if}\  i_1 < r_1 \mathtt{then}\ o:=q_1,\\
\mathtt{else if}\  i_2 < r_2\ \mathtt{then}\ o:=q_2,\\
\mathtt{else }\  o:=q_0,
\end{align}
and if the mode is $m_2$ then
\begin{align}
\label{eq:controllernew2}
\mathtt{if}\  i_1 < r_1 \land  i_2< r_2 \ \mathtt{then}\ o:=q_2, m:=m_1,\\
\mathtt{else if}\  i_1 < r_1 \mathtt{then}\ o:=q_1,\\
\mathtt{else if}\  i_2 < r_2\ \mathtt{then}\ o:=q_2,\\
\mathtt{else }\  o:=q_0.
\end{align}
If the new controller sees that both tanks are below the thresholds $r_1$ and $r_2$, it will not always prioritise the filling of $T_1$ as it did before. We informally justify that the controllability of $\alpha_3$ over the variable $y_2$ changed, as it is now not possible for the attacker to empty $T_2$. The tank $T_2$ does not run out of water, because the controller will fill $T_2$ when it sees that the water level of $T_2$ is below $r_2$ even if the attacker always tricks the controller into thinking that $T_1$ is empty, a behaviour that was not present before. However, if the water level of tank $T_2$ is at its minimum normal value, $r_2^-$, it could happen that the controller's mode is $m_1$ at that time, and it will take an extra cycle to switch from mode $m_1$ to $m_2$. During that extra cycle, the tank $T_2$ loses $v_2$ units of water, so the attacker's control over $y_2$ is lower bounded by $r_2^--v_2$, as shown in Table \ref{tab:AttackersNew}. Consequently, the attacker $\alpha_3$ cannot make the system reach a state where the tank $T_2$ is empty.
{\small
\begin{table}[!h]
\centering
\begin{tabular}{|c|c|c|c|c|c|c|c|}
\hline
 \multirow{2}{*}{Attacker} &  \multirow{2}{*}{Controls} &  \multicolumn{2}{c|}{Quantification} & \multicolumn{4}{c|}{Vulnerable?}\\\cline{3-8}
 &	&	$y_1$& $y_2$& $E_1$&	$E_2$	&	$F_1$&   $F_2$\\
\hline\hline
$\alpha_3$	&	${i_1}$ &$[0,L]$&$[r_2^{-}-v_2,r_2^+]$ & $\cmark$&$\xmark$&$\cmark$&$\xmark$ \\
\hline
\end{tabular}
\caption{Quantification of the $k$-controllability of attacker $\alpha_3$ with the new controller in place.} 
\label{tab:AttackersNew}
\vspace{-10pt}
\end{table}
}
\section{Discussion and Future Work}
\label{sec:Discussion}
In this section, we list and discuss some fundamental discussion points, some of which represent the basis for interesting future work. 

\subsection{Modelling and Abstraction level}
Although we provide a formal framework for the modelling of CPSs, their states, and their execution semantics, and showed evidence that it can be used to derive useful analysis, it will be a subject of future work to assess how appropriate our modelling is when applied to other, perhaps more complex, scenarios. In essence, our formalism defines a deterministic transition system whose transition function is given by the one-step cycle semantics function $\TheSystem$ (see Definition \ref{def:SingleCycleSemantics}). Consequently, the notion of integrity that we use in Definition \ref{def:NICPS} is not probabilistic. However, we believe that this approach is good enough to carry out an approximate analysis if the probabilistic nature of the process is due to the existence of zero-mean noises.

\emph{Hybrid automata} \cite{ALUR19953} are state-based models for dynamical systems. These automata allow the description of invariants and side-effects on transitions, making them suitable for the description of many continuous physical processes. There are also model checking techniques for hybrid automata. It would be interesting to see if these automata provide any advantages when it comes to automatic verification of the proposed properties.


\subsection{Quantification of Interference}
For our case study, we came up with a notion of quantification of integrity violations which, ultimately, was related to how many unwanted states could an attacker force the system to be in, in other words the level of \emph{corruption} an attacker can induce. These metrics are in a sense dual to metrics used in quantitative information flow, which measure the entropy of information, and how likely it is for an attacker to guess a secret based on the information she has already observed. 

Ideally, we want to quantify the $k$-controllability of attackers for large values of $k$ in order to cover as many possible states. Unfortunately, the larger $k$ is, the more computationally demanding the analysis becomes. However, we want to highlight that the behaviour of CPSs is usually regular; \emph{i.e.}, the operation modes of CPSs often follow the same sequence over and over, revisiting similar states each time the process repeats. Due to this regular nature, it may be possible to restrict the analysis to the length of this ``behavioural loop'' instead of analysing arbitrarily long traces, therefore considering a well defined and relatively small $\bar{k} \in \mathbb{N}$.

\subsection{A Theory of Attacks and Attackers}
From our case study, we observe that it is possible for different attackers to drive the CPS to critical states by using different attack strategies. Thus, it may be convenient to define more abstract notions of attack and attacker that are more related to the effects that we want to avoid on the system in order to bundle together attacks and attackers that are equally powerful.
While we think that our attacker model is at least as powerful as attacker models commonly used in control theory (see Section \ref{sec:AttacksOnModels}), further work is indeed to formally compare them. 

\subsection{Properties}
There are many well-known properties in control theory that characterise important aspects of dynamical systems, including {controllability, observability, stability} and {stabilizability}. 
Since we focus on the protection of the integrity of the process, we want to reduce the degree of \emph{controllability} that the attacker has over the system. A focus on the protection of the confidentiality of data would aim to reduce the degree of \emph{observability} that the attacker has over the system. Krotofil and Larsen acknowledge the importance of reasoning about {controllability} and {observability}, but they highlight that it is also important to reason about \emph{operability}, which is ``the ability to achieve acceptable operations'' \cite{krotofil2015rocking}. Determining how our notion of process integrity relates to these other properties is an interesting research topic.


\subsection{Proof Automation}
From the manual analysis of the case study CPS in the previous seciton, we foresee that it will be possible to use existing techniques (such as automated theorem proving, model-checking and SMT solving) to automate parts of the verification. The definition of the measures for quantification, the notions of critical state, the attackers and the initial states must naturally be done manually, but the problem of verification can be reduced to the quantification of the reachable states by an attacker (similar as in quantitative information flow analysis for confidentiality). 
It would be also interesting to determine whether our approach together with state-of-the-art verification tools has advantages over existing control-theoretical reasoning techniques, which focus on solving algebraic 
representations of CPSs.

\subsection{Redesign}
In our case study, we determined that the attacker $\alpha_3$ was more powerful than the attacker $\alpha_2$ due to the way the controller was designed. Thus, we believe that it would be interesting to develop a methodology for (re)designing controllers that helps avoiding these type of unsafe behaviours. 

\section{Related Work}

There is a wealth of work related to the modelling and verification of Cyber-Physical Systems, mostly with focus on traditional \emph{safety} properties: both in the formal sense (properties of a single trace) and the 
informal sense (resilience against random faults), see for instance \cite{alur2015principles} for a good survey on this field. In the following thus we focus mostly on formal models of CPSs security.

The applicability of Information Flow Analysis (IFA) to CPS security has been reinforced by several authors, although many of the works do not provide definitive results, and other authors focus on protecting confidentiality and not integrity.
{
}

In \cite{CPSSec}, the Gollmann and Krotofil state that ``Physical relationships between the variables in an industrial process, e.g. between volume, pressure, and temperature in a vessel, can be viewed as information flows from a modelling perspective,'' acknowledging that it is possible to model physical aspects of CPS using an IF setting; however, they do not explicitly say how. Gamage \emph{et al.} \cite{Gamage2010} use IFA to illustrate how to prevent attacks on confidentiality of CPS though the notion of \emph{compensating pair} $(a, a^c)$, where $a^c$ is an action that cancels the physical manifestation of the earlier occurring action $a$ so that attackers do not infer that $a$ took place.

{
Other authors have proposed attacker models for CPSs. For example, Howser and McMillin provide in \cite{StuxnetOnCPS} an IF-based attacker model that builds on top of nondeducibility \cite{Nondeducibility}, where attackers aim to hide information relevant to attacks or faults to the monitoring systems, preventing operators from realising that the behaviour of the system is anomalous. Given that we our ultimate goal is to make CPSs more resilient by design, we consider all possible behaviours that attackers could trigger, so our analysis is not aimed at deciding whether the attacker hides information from the operator. Rocchetto and Tippenhauer \cite{CPSDolevYao} extend the Dolev-Yao attacker model \cite{DolevYao} to the Cyber-Physical Dolev-Yao (CPDY) model, where attackers can interact with the physical domain through orthogonal channels. We believe that their attackers can be modelled in our framework with  attackers that control the components of the control vector $u$, but a formal justification is still missing.
}
From the perspective of control theory, the work by Weerakkody \emph{et al.} \cite{IFCPSSec} proposes the \emph{Kullback-Leibler divergence} between the distributions of the attacked and attack-free residuals as {a measure for information flow} to determine to which extent the actions of an attacker interfere with the system. However, their definition of security is ultimately tied to a maximum deviation from a prediction model, while ours is based on semantics-based information flow. Thus, although we can imagine that both measures are somehow related, it is not evident what their exact relationship is. In the work by Murguia \emph{et al.} \cite{ReachableSets}, the authors characterise the states of the system that can be reached when under attack, and they try to minimise this set of reachable states by modifying the mathematical model of the controller. Their approach to characterise the security of the system is purely control-theoretical, while ours is based on information-flow analysis.

Finally, verification tools like McLaughlin \emph{et al.}'s \cite{TSVPLC} that use symbolic execution for model checking safety properties in PLCs could help our approach by helping to quantify the interference of attackers.
{
In sum, to the best of our knowledge, we are the first to propose the use of a semantics-based approach, inspired by traditional information-flow analysis techniques used for software security, to quantify the impact of attacker actions on process variables at design time in CPSs.
}


\section{Conclusions}
\label{sec:Conclusion}
In this work we 
proposed a formal technique to reason about the security of Cyber-Physical Systems based on information-flow analysis and focused on integrity properties. 
Our basic observation is that in CPS security often the worst case scenario is 
related with the integrity of some physical state (i.e. level of water or 
chemical concentration in tank, temperature etc.) and that an attacker's goal 
is to reach that state by manipulating certain digital or physical inputs to 
a system (by tampering with one or more sensor or actuators). As such, we can 
cast this problem as an information flow problem and leverage on well-known 
principles to perform this analysis. We have illustrated our approach by 
means of a realistic case study and showed that we can identify and quantify 
non-trivial harmful flows.

In the future we plan to perform this reasoning in a semi-automatic fashion by 
leveraging on model-checking technology, and we plan to apply our methodology 
to a wider range of CPS from various domains (such as electricity, water 
treatment).

\begin{acks}
We would like to thank Bruce McMillin for his insightful comments and suggestions to earlier versions of this draft.
\end{acks}
\bibliographystyle{ACM-Reference-Format}
\bibliography{SIMEi}


\begin{thebibliography}{00}


\ifx \showCODEN    \undefined \def \showCODEN     #1{\unskip}     \fi
\ifx \showDOI      \undefined \def \showDOI       #1{#1}\fi
\ifx \showISBNx    \undefined \def \showISBNx     #1{\unskip}     \fi
\ifx \showISBNxiii \undefined \def \showISBNxiii  #1{\unskip}     \fi
\ifx \showISSN     \undefined \def \showISSN      #1{\unskip}     \fi
\ifx \showLCCN     \undefined \def \showLCCN      #1{\unskip}     \fi
\ifx \shownote     \undefined \def \shownote      #1{#1}          \fi
\ifx \showarticletitle \undefined \def \showarticletitle #1{#1}   \fi
\ifx \showURL      \undefined \def \showURL       {\relax}        \fi
\providecommand\bibfield[2]{#2}
\providecommand\bibinfo[2]{#2}
\providecommand\natexlab[1]{#1}
\providecommand\showeprint[2][]{arXiv:#2}

\bibitem[\protect\citeauthoryear{??}{Stu}{2011}]%
        {StuxnetWeb}
 \bibinfo{year}{2011}\natexlab{}.
\newblock \bibinfo{title}{{The Man Who Found Stuxnet -- Sergey Ulasen in the
  Spotlight}}.
\newblock
  \bibinfo{howpublished}{\url{https://eugene.kaspersky.com/2011/11/02/the-man-who-found-stuxnet-sergey-ulasen-in-the-spotlight/}}.
    (\bibinfo{year}{2011}).
\newblock
\newblock
\shownote{{Accessed: 2017-01-09}.}


\bibitem[\protect\citeauthoryear{??}{Lag}{2014}]%
        {Lagebericht2014}
 \bibinfo{year}{2014}\natexlab{}.
\newblock \bibinfo{title}{{Bericht zur Lage der IT-Sicherheit in Deutschland
  2014}}.
\newblock
  \bibinfo{howpublished}{\url{https://www.bsi.bund.de/SharedDocs/Downloads/DE/BSI/Publikationen/\\Lageberichte/Lagebericht2014.pdf}}.
    (\bibinfo{year}{2014}).
\newblock
\newblock
\shownote{{In German, Accessed: 2017-01-09}.}


\bibitem[\protect\citeauthoryear{??}{Wir}{2015}]%
        {WiredArticle}
 \bibinfo{year}{2015}\natexlab{}.
\newblock \bibinfo{title}{{A Cyberattack Has Caused Confirmed Physical Damage
  for the Second Time Ever}}.
\newblock
  \bibinfo{howpublished}{\url{https://www.wired.com/2015/01/german-steel-mill-hack-destruction/}}.
    (\bibinfo{year}{2015}).
\newblock
\newblock
\shownote{{Accessed: 2017-01-09}.}


\bibitem[\protect\citeauthoryear{Adams, Woodall, and Lowry}{Adams
  et~al\mbox{.}}{1992}]%
        {Adams}
\bibfield{author}{\bibinfo{person}{B.M. Adams}, \bibinfo{person}{W.H. Woodall},
  {and} \bibinfo{person}{C.A. Lowry}.} \bibinfo{year}{1992}\natexlab{}.
\newblock \showarticletitle{The use (and misuse) of false alarm probabilities
  in control chart design}.
\newblock \bibinfo{journal}{{\em Frontiers in Statistical Quality Control 4\/}}
  (\bibinfo{year}{1992}), \bibinfo{pages}{155--168}.
\newblock


\bibitem[\protect\citeauthoryear{Adepu and Mathur}{Adepu and Mathur}{2016}]%
        {CPSInvariantsForDetection}
\bibfield{author}{\bibinfo{person}{Sridhar Adepu} {and} \bibinfo{person}{Aditya
  Mathur}.} \bibinfo{year}{2016}\natexlab{}.
\newblock \bibinfo{booktitle}{{\em Using Process Invariants to Detect Cyber
  Attacks on a Water Treatment System}}.
\newblock \bibinfo{publisher}{Springer International Publishing},
  \bibinfo{address}{Cham}, \bibinfo{pages}{91--104}.
\newblock
\showISBNx{978-3-319-33630-5}
\showDOI{%
\url{https://doi.org/10.1007/978-3-319-33630-5\_7}}


\bibitem[\protect\citeauthoryear{Alur}{Alur}{2015}]%
        {alur2015principles}
\bibfield{author}{\bibinfo{person}{Rajeev Alur}.}
  \bibinfo{year}{2015}\natexlab{}.
\newblock \bibinfo{booktitle}{{\em Principles of cyber-physical systems}}.
\newblock \bibinfo{publisher}{MIT Press}.
\newblock


\bibitem[\protect\citeauthoryear{Alur, Courcoubetis, Halbwachs, Henzinger, Ho,
  Nicollin, Olivero, Sifakis, and Yovine}{Alur et~al\mbox{.}}{1995}]%
        {ALUR19953}
\bibfield{author}{\bibinfo{person}{R. Alur}, \bibinfo{person}{C. Courcoubetis},
  \bibinfo{person}{N. Halbwachs}, \bibinfo{person}{T.A. Henzinger},
  \bibinfo{person}{P.-H. Ho}, \bibinfo{person}{X. Nicollin},
  \bibinfo{person}{A. Olivero}, \bibinfo{person}{J. Sifakis}, {and}
  \bibinfo{person}{S. Yovine}.} \bibinfo{year}{1995}\natexlab{}.
\newblock \showarticletitle{The algorithmic analysis of hybrid systems}.
\newblock \bibinfo{journal}{{\em Theoretical Computer Science\/}}
  \bibinfo{volume}{138}, \bibinfo{number}{1} (\bibinfo{year}{1995}),
  \bibinfo{pages}{3 -- 34}.
\newblock
\showISSN{0304-3975}
\showDOI{%
\url{https://doi.org/10.1016/0304-3975(94)00202-T}}
\newblock
\shownote{Hybrid Systems.}


\bibitem[\protect\citeauthoryear{Astr\"{o}m and Wittenmark}{Astr\"{o}m and
  Wittenmark}{1997}]%
        {Astrom}
\bibfield{author}{\bibinfo{person}{Karl~J. Astr\"{o}m} {and}
  \bibinfo{person}{Bj\"{o}rn Wittenmark}.} \bibinfo{year}{1997}\natexlab{}.
\newblock \bibinfo{booktitle}{{\em Computer-controlled Systems (3rd Ed.)}}.
\newblock \bibinfo{publisher}{Prentice-Hall, Inc.}, \bibinfo{address}{Upper
  Saddle River, NJ, USA}.
\newblock


\bibitem[\protect\citeauthoryear{Bai, Pasqualetti, and Gupta}{Bai
  et~al\mbox{.}}{2015}]%
        {Gupta2}
\bibfield{author}{\bibinfo{person}{C.~Z. Bai}, \bibinfo{person}{F.
  Pasqualetti}, {and} \bibinfo{person}{V. Gupta}.}
  \bibinfo{year}{2015}\natexlab{}.
\newblock \showarticletitle{Security in stochastic control systems: Fundamental
  limitations and performance bounds}. In \bibinfo{booktitle}{{\em American
  Control Conference (ACC), 2015}}. \bibinfo{pages}{195--200}.
\newblock


\bibitem[\protect\citeauthoryear{Basseville}{Basseville}{1988}]%
        {Basseville}
\bibfield{author}{\bibinfo{person}{M. Basseville}.}
  \bibinfo{year}{1988}\natexlab{}.
\newblock \showarticletitle{Detecting changes in signals and systems - a
  survey}.
\newblock \bibinfo{journal}{{\em Automatica\/}}  \bibinfo{volume}{24}
  (\bibinfo{year}{1988}), \bibinfo{pages}{309 -- 326}.
\newblock


\bibitem[\protect\citeauthoryear{Bell and Padula}{Bell and Padula}{1973}]%
        {BellLapadula}
\bibfield{author}{\bibinfo{person}{D.E. Bell} {and} \bibinfo{person}{L.J.L.
  Padula}.} \bibinfo{year}{1973}\natexlab{}.
\newblock \bibinfo{booktitle}{{\em Secure Computer Systems: Mathematical
  Foundations and Model}}.
\newblock Number v. 1. \bibinfo{publisher}{Mitre Corp.}
\newblock


\bibitem[\protect\citeauthoryear{Biba}{Biba}{1977}]%
        {BibaIntegrity}
\bibfield{author}{\bibinfo{person}{Kenneth~J Biba}.}
  \bibinfo{year}{1977}\natexlab{}.
\newblock \bibinfo{booktitle}{{\em Integrity considerations for secure computer
  systems}}.
\newblock \bibinfo{type}{{T}echnical {R}eport} MTR-3153.
  \bibinfo{institution}{MITRE Corporation}.
\newblock


\bibitem[\protect\citeauthoryear{C\'{a}rdenas, Amin, Lin, Huang, Huang, and
  Sastry}{C\'{a}rdenas et~al\mbox{.}}{2011}]%
        {CPSAttacksAgainstPCS}
\bibfield{author}{\bibinfo{person}{Alvaro~A. C\'{a}rdenas},
  \bibinfo{person}{Saurabh Amin}, \bibinfo{person}{Zong-Syun Lin},
  \bibinfo{person}{Yu-Lun Huang}, \bibinfo{person}{Chi-Yen Huang}, {and}
  \bibinfo{person}{Shankar Sastry}.} \bibinfo{year}{2011}\natexlab{}.
\newblock \showarticletitle{Attacks Against Process Control Systems: Risk
  Assessment, Detection, and Response}. In \bibinfo{booktitle}{{\em Proceedings
  of the 6th ACM Symposium on Information, Computer and Communications
  Security}} {\em (\bibinfo{series}{ASIACCS '11})}. \bibinfo{publisher}{ACM},
  \bibinfo{address}{New York, NY, USA}, \bibinfo{pages}{355--366}.
\newblock
\showISBNx{978-1-4503-0564-8}
\showDOI{%
\url{https://doi.org/10.1145/1966913.1966959}}


\bibitem[\protect\citeauthoryear{Chen and Patton}{Chen and Patton}{1999}]%
        {Patton_1}
\bibfield{author}{\bibinfo{person}{Jie Chen} {and} \bibinfo{person}{Ron~J.
  Patton}.} \bibinfo{year}{1999}\natexlab{}.
\newblock \bibinfo{booktitle}{{\em Robust Model-based Fault Diagnosis for
  Dynamic Systems}}.
\newblock \bibinfo{publisher}{Kluwer Academic Publishers},
  \bibinfo{address}{Norwell, MA, USA}.
\newblock


\bibitem[\protect\citeauthoryear{Clark, Hunt, and Malacaria}{Clark
  et~al\mbox{.}}{2002}]%
        {QIF}
\bibfield{author}{\bibinfo{person}{David Clark}, \bibinfo{person}{Sebastian
  Hunt}, {and} \bibinfo{person}{Pasquale Malacaria}.}
  \bibinfo{year}{2002}\natexlab{}.
\newblock \showarticletitle{{Quantitative Analysis of the Leakage of
  Confidential Data}}.
\newblock \bibinfo{journal}{{\em Electronic Notes in Theoretical Computer
  Science\/}} \bibinfo{volume}{59}, \bibinfo{number}{3} (\bibinfo{year}{2002}).
\newblock


\bibitem[\protect\citeauthoryear{Clark and Wilson}{Clark and Wilson}{1987}]%
        {ClarkWilson87}
\bibfield{author}{\bibinfo{person}{D.~D. Clark} {and} \bibinfo{person}{D.~R.
  Wilson}.} \bibinfo{year}{1987}\natexlab{}.
\newblock \showarticletitle{A Comparison of Commercial and Military Computer
  Security Policies}. In \bibinfo{booktitle}{{\em 1987 IEEE Symposium on
  Security and Privacy}}. \bibinfo{pages}{184--194}.
\newblock
\showISSN{1540-7993}
\showDOI{%
\url{https://doi.org/10.1109/SP.1987.10001}}


\bibitem[\protect\citeauthoryear{Clarkson and Schneider}{Clarkson and
  Schneider}{2010a}]%
        {Hyperproperties}
\bibfield{author}{\bibinfo{person}{Michael~R. Clarkson} {and}
  \bibinfo{person}{Fred~B. Schneider}.} \bibinfo{year}{2010}\natexlab{a}.
\newblock \showarticletitle{Hyperproperties}.
\newblock \bibinfo{journal}{{\em Journal of Computer Security\/}}
  \bibinfo{volume}{18}, \bibinfo{number}{6} (\bibinfo{year}{2010}),
  \bibinfo{pages}{1157--1210}.
\newblock
\showDOI{%
\url{https://doi.org/10.3233/JCS-2009-0393}}


\bibitem[\protect\citeauthoryear{Clarkson and Schneider}{Clarkson and
  Schneider}{2010b}]%
        {QuantitativeIntegrity}
\bibfield{author}{\bibinfo{person}{M.~R. Clarkson} {and} \bibinfo{person}{F.~B.
  Schneider}.} \bibinfo{year}{2010}\natexlab{b}.
\newblock \showarticletitle{Quantification of Integrity}. In
  \bibinfo{booktitle}{{\em 2010 23rd IEEE Computer Security Foundations
  Symposium}}. \bibinfo{pages}{28--43}.
\newblock
\showISSN{1063-6900}
\showDOI{%
\url{https://doi.org/10.1109/CSF.2010.10}}


\bibitem[\protect\citeauthoryear{Dan and Sandberg}{Dan and Sandberg}{2010}]%
        {CPSStealthAttacks}
\bibfield{author}{\bibinfo{person}{G. Dan} {and} \bibinfo{person}{H.
  Sandberg}.} \bibinfo{year}{2010}\natexlab{}.
\newblock \showarticletitle{Stealth Attacks and Protection Schemes for State
  Estimators in Power Systems}. In \bibinfo{booktitle}{{\em 2010 First IEEE
  International Conference on Smart Grid Communications}}.
  \bibinfo{pages}{214--219}.
\newblock
\showDOI{%
\url{https://doi.org/10.1109/SMARTGRID.2010.5622046}}


\bibitem[\protect\citeauthoryear{Denning}{Denning}{1975}]%
        {SecureInformationFlows}
\bibfield{author}{\bibinfo{person}{Dorothy Elizabeth~Robling Denning}.}
  \bibinfo{year}{1975}\natexlab{}.
\newblock {\em \bibinfo{title}{Secure Information Flow in Computer Systems.}}
\newblock \bibinfo{thesistype}{Ph.D. Dissertation}. \bibinfo{address}{West
  Lafayette, IN, USA}.
\newblock
\newblock
\shownote{AAI7600514.}


\bibitem[\protect\citeauthoryear{Dolev and Yao}{Dolev and Yao}{1983}]%
        {DolevYao}
\bibfield{author}{\bibinfo{person}{D. Dolev} {and} \bibinfo{person}{A. Yao}.}
  \bibinfo{year}{1983}\natexlab{}.
\newblock \showarticletitle{On the security of public key protocols}.
\newblock \bibinfo{journal}{{\em IEEE Transactions on Information Theory\/}}
  \bibinfo{volume}{29}, \bibinfo{number}{2} (\bibinfo{date}{Mar}
  \bibinfo{year}{1983}), \bibinfo{pages}{198--208}.
\newblock
\showISSN{0018-9448}
\showDOI{%
\url{https://doi.org/10.1109/TIT.1983.1056650}}


\bibitem[\protect\citeauthoryear{Gamage, McMillin, and Roth}{Gamage
  et~al\mbox{.}}{2010}]%
        {Gamage2010}
\bibfield{author}{\bibinfo{person}{T.~T. Gamage}, \bibinfo{person}{B.~M.
  McMillin}, {and} \bibinfo{person}{T.~P. Roth}.}
  \bibinfo{year}{2010}\natexlab{}.
\newblock \showarticletitle{Enforcing Information Flow Security Properties in
  Cyber-Physical Systems: A Generalized Framework Based on Compensation}. In
  \bibinfo{booktitle}{{\em 2010 IEEE 34th Annual Computer Software and
  Applications Conference Workshops}}. \bibinfo{pages}{158--163}.
\newblock
\showDOI{%
\url{https://doi.org/10.1109/COMPSACW.2010.36}}


\bibitem[\protect\citeauthoryear{Gertler}{Gertler}{1988}]%
        {Gertler}
\bibfield{author}{\bibinfo{person}{J. Gertler}.}
  \bibinfo{year}{1988}\natexlab{}.
\newblock \showarticletitle{Survey of model-based failure detection and
  isolation in complex plants}.
\newblock \bibinfo{journal}{{\em Control Systems Magazine, IEEE\/}}
  \bibinfo{volume}{8} (\bibinfo{year}{1988}), \bibinfo{pages}{3--11}.
\newblock


\bibitem[\protect\citeauthoryear{Goguen and Meseguer}{Goguen and
  Meseguer}{1982}]%
        {Noninterference}
\bibfield{author}{\bibinfo{person}{Joseph~A. Goguen} {and}
  \bibinfo{person}{Jos{\'{e}} Meseguer}.} \bibinfo{year}{1982}\natexlab{}.
\newblock \showarticletitle{Security Policies and Security Models}. In
  \bibinfo{booktitle}{{\em 1982 {IEEE} Symposium on Security and Privacy,
  Oakland, CA, USA, April 26-28, 1982}}. \bibinfo{pages}{11--20}.
\newblock
\showDOI{%
\url{https://doi.org/10.1109/SP.1982.10014}}


\bibitem[\protect\citeauthoryear{Gollmann, Gurikov, Isakov, Krotofil, Larsen,
  and Winnicki}{Gollmann et~al\mbox{.}}{2015}]%
        {CPSSecVinyl}
\bibfield{author}{\bibinfo{person}{Dieter Gollmann}, \bibinfo{person}{Pavel
  Gurikov}, \bibinfo{person}{Alexander Isakov}, \bibinfo{person}{Marina
  Krotofil}, \bibinfo{person}{Jason Larsen}, {and} \bibinfo{person}{Alexander
  Winnicki}.} \bibinfo{year}{2015}\natexlab{}.
\newblock \showarticletitle{Cyber-Physical Systems Security: Experimental
  Analysis of a Vinyl Acetate Monomer Plant}. In \bibinfo{booktitle}{{\em
  Proceedings of the 1st ACM Workshop on Cyber-Physical System Security}} {\em
  (\bibinfo{series}{CPSS '15})}. \bibinfo{publisher}{ACM},
  \bibinfo{address}{New York, NY, USA}, \bibinfo{pages}{1--12}.
\newblock
\showISBNx{978-1-4503-3448-8}
\showDOI{%
\url{https://doi.org/10.1145/2732198.2732208}}


\bibitem[\protect\citeauthoryear{Gollmann and Krotofil}{Gollmann and
  Krotofil}{2016}]%
        {CPSSec}
\bibfield{author}{\bibinfo{person}{Dieter Gollmann} {and}
  \bibinfo{person}{Marina Krotofil}.} \bibinfo{year}{2016}\natexlab{}.
\newblock \bibinfo{booktitle}{{\em Cyber-Physical Systems Security}}.
\newblock \bibinfo{publisher}{Springer Berlin Heidelberg},
  \bibinfo{address}{Berlin, Heidelberg}, \bibinfo{pages}{195--204}.
\newblock
\showISBNx{978-3-662-49301-4}
\showDOI{%
\url{https://doi.org/10.1007/978-3-662-49301-4\_14}}


\bibitem[\protect\citeauthoryear{Gustafsson}{Gustafsson}{2000}]%
        {Gustafsson}
\bibfield{author}{\bibinfo{person}{F. Gustafsson}.}
  \bibinfo{year}{2000}\natexlab{}.
\newblock \bibinfo{booktitle}{{\em Adaptive Filtering and Change Detection}}.
\newblock \bibinfo{publisher}{John Wiley and Sons, LTD}, \bibinfo{address}{West
  Sussex, Chichester, England}.
\newblock


\bibitem[\protect\citeauthoryear{Heemels and Schutter}{Heemels and
  Schutter}{2007}]%
        {ModelingAndControlOfHybridDynamicalSystems}
\bibfield{author}{\bibinfo{person}{Maurice Heemels} {and}
  \bibinfo{person}{Bart~De Schutter}.} \bibinfo{year}{2007}\natexlab{}.
\newblock \bibinfo{title}{Modeling and Control of Hybrid Dynamical Systems}.
\newblock Einhoven University of Technology.
\newblock


\bibitem[\protect\citeauthoryear{Howser and McMillin}{Howser and
  McMillin}{2014}]%
        {StuxnetOnCPS}
\bibfield{author}{\bibinfo{person}{G. Howser} {and} \bibinfo{person}{B.
  McMillin}.} \bibinfo{year}{2014}\natexlab{}.
\newblock \showarticletitle{A Modal Model of Stuxnet Attacks on Cyber-physical
  Systems: A Matter of Trust}. In \bibinfo{booktitle}{{\em 2014 Eighth
  International Conference on Software Security and Reliability (SERE)}}.
  \bibinfo{pages}{225--234}.
\newblock
\showDOI{%
\url{https://doi.org/10.1109/SERE.2014.36}}


\bibitem[\protect\citeauthoryear{Krotofil and Larsen}{Krotofil and
  Larsen}{2015}]%
        {krotofil2015rocking}
\bibfield{author}{\bibinfo{person}{Marina Krotofil} {and}
  \bibinfo{person}{Jason Larsen}.} \bibinfo{year}{2015}\natexlab{}.
\newblock \showarticletitle{Rocking the pocket book: Hacking chemical plants}.
  In \bibinfo{booktitle}{{\em DefCon Conference, DEFCON}}.
\newblock


\bibitem[\protect\citeauthoryear{Kyriakides and Polycarpou}{Kyriakides and
  Polycarpou}{2015}]%
        {Marios_Poly}
\bibfield{editor}{\bibinfo{person}{Elias Kyriakides} {and}
  \bibinfo{person}{Marios~M. Polycarpou}} (Eds.).
  \bibinfo{year}{2015}\natexlab{}.
\newblock \bibinfo{booktitle}{{\em Intelligent Monitoring, Control, and
  Security of Critical Infrastructure Systems}}. \bibinfo{series}{Studies in
  Computational Intelligence}, Vol.~\bibinfo{volume}{565}.
\newblock \bibinfo{publisher}{Springer}.
\newblock


\bibitem[\protect\citeauthoryear{Langner}{Langner}{2011}]%
        {Stuxnet}
\bibfield{author}{\bibinfo{person}{R. Langner}.}
  \bibinfo{year}{2011}\natexlab{}.
\newblock \showarticletitle{Stuxnet: Dissecting a Cyberwarfare Weapon}.
\newblock \bibinfo{journal}{{\em IEEE Security Privacy\/}} \bibinfo{volume}{9},
  \bibinfo{number}{3} (\bibinfo{date}{May} \bibinfo{year}{2011}),
  \bibinfo{pages}{49--51}.
\newblock
\showISSN{1540-7993}
\showDOI{%
\url{https://doi.org/10.1109/MSP.2011.67}}


\bibitem[\protect\citeauthoryear{Luenberger}{Luenberger}{1966}]%
        {1098323}
\bibfield{author}{\bibinfo{person}{D. Luenberger}.}
  \bibinfo{year}{1966}\natexlab{}.
\newblock \showarticletitle{Observers for multivariable systems}.
\newblock \bibinfo{journal}{{\it IEEE Trans. Automat. Control}}
  \bibinfo{volume}{11} (\bibinfo{year}{1966}), \bibinfo{pages}{190--197}.
\newblock


\bibitem[\protect\citeauthoryear{McLaughlin, Zonouz, Pohly, and
  McDaniel}{McLaughlin et~al\mbox{.}}{2014}]%
        {TSVPLC}
\bibfield{author}{\bibinfo{person}{Stephen~E. McLaughlin},
  \bibinfo{person}{Saman~A. Zonouz}, \bibinfo{person}{Devin~J. Pohly}, {and}
  \bibinfo{person}{Patrick~D. McDaniel}.} \bibinfo{year}{2014}\natexlab{}.
\newblock \showarticletitle{A Trusted Safety Verifier for Process Controller
  Code}. In \bibinfo{booktitle}{{\em 21st Annual Network and Distributed System
  Security Symposium, {NDSS} 2014, San Diego, California, USA, February 23-26,
  2014}}.
\newblock
\showURL{%
\url{http://www.internetsociety.org/doc/trusted-safety-verifier-process-controller-code}}


\bibitem[\protect\citeauthoryear{Mo, Chabukswar, and Sinopoli}{Mo
  et~al\mbox{.}}{2014}]%
        {CPSDetectingIntegrityAttacksScada}
\bibfield{author}{\bibinfo{person}{Y. Mo}, \bibinfo{person}{R. Chabukswar},
  {and} \bibinfo{person}{B. Sinopoli}.} \bibinfo{year}{2014}\natexlab{}.
\newblock \showarticletitle{Detecting Integrity Attacks on SCADA Systems}.
\newblock \bibinfo{journal}{{\em IEEE Transactions on Control Systems
  Technology\/}} \bibinfo{volume}{22}, \bibinfo{number}{4}
  (\bibinfo{date}{July} \bibinfo{year}{2014}), \bibinfo{pages}{1396--1407}.
\newblock
\showISSN{1063-6536}
\showDOI{%
\url{https://doi.org/10.1109/TCST.2013.2280899}}


\bibitem[\protect\citeauthoryear{Mo and Sinopoli}{Mo and Sinopoli}{2009}]%
        {CPSReplayAttacks}
\bibfield{author}{\bibinfo{person}{Y. Mo} {and} \bibinfo{person}{B. Sinopoli}.}
  \bibinfo{year}{2009}\natexlab{}.
\newblock \showarticletitle{Secure control against replay attacks}. In
  \bibinfo{booktitle}{{\em 2009 47th Annual Allerton Conference on
  Communication, Control, and Computing (Allerton)}}.
  \bibinfo{pages}{911--918}.
\newblock
\showDOI{%
\url{https://doi.org/10.1109/ALLERTON.2009.5394956}}


\bibitem[\protect\citeauthoryear{Mo and Sinopoli}{Mo and Sinopoli}{2010}]%
        {CPSDataInjectionAttacks}
\bibfield{author}{\bibinfo{person}{Y. Mo} {and} \bibinfo{person}{B. Sinopoli}.}
  \bibinfo{year}{2010}\natexlab{}.
\newblock \showarticletitle{False data injection attacks in control systems}.
  In \bibinfo{booktitle}{{\em Proc. 1st Workshop Secure Control Syst.,
  Stocholm, Sweeden}}.
\newblock


\bibitem[\protect\citeauthoryear{Mo and Sinopoli}{Mo and Sinopoli}{2012}]%
        {CPSIntegrityAttacks}
\bibfield{author}{\bibinfo{person}{Yilin Mo} {and} \bibinfo{person}{Bruno
  Sinopoli}.} \bibinfo{year}{2012}\natexlab{}.
\newblock \showarticletitle{Integrity Attacks on Cyber-physical Systems}. In
  \bibinfo{booktitle}{{\em Proceedings of the 1st International Conference on
  High Confidence Networked Systems}} {\em (\bibinfo{series}{HiCoNS '12})}.
  \bibinfo{publisher}{ACM}, \bibinfo{address}{New York, NY, USA},
  \bibinfo{pages}{47--54}.
\newblock
\showISBNx{978-1-4503-1263-9}
\showDOI{%
\url{https://doi.org/10.1145/2185505.2185514}}


\bibitem[\protect\citeauthoryear{Murguia and Ruths}{Murguia and Ruths}{2016a}]%
        {Carlos_Justin1}
\bibfield{author}{\bibinfo{person}{Carlos Murguia} {and}
  \bibinfo{person}{Justin Ruths}.} \bibinfo{year}{2016}\natexlab{a}.
\newblock \showarticletitle{Characterization of a CUSUM Model-Based Sensor
  Attack Detector}. In \bibinfo{booktitle}{{\em proceedings of the 55th IEEE
  Conference on Decision and Control (CDC)}}.
\newblock


\bibitem[\protect\citeauthoryear{Murguia and Ruths}{Murguia and Ruths}{2016b}]%
        {Carlos_Justin2}
\bibfield{author}{\bibinfo{person}{Carlos Murguia} {and}
  \bibinfo{person}{Justin Ruths}.} \bibinfo{year}{2016}\natexlab{b}.
\newblock \showarticletitle{CUSUM and Chi-Squared Attack Detection of
  Compromised Sensors}. In \bibinfo{booktitle}{{\em proceedings of the IEEE
  Multi-Conference on Systems and Control (MSC)}}.
\newblock


\bibitem[\protect\citeauthoryear{Murguia, van~de Wouw, and Ruths}{Murguia
  et~al\mbox{.}}{2017}]%
        {ReachableSets}
\bibfield{author}{\bibinfo{person}{Carlos Murguia}, \bibinfo{person}{Nathan
  van~de Wouw}, {and} \bibinfo{person}{Justin Ruths}.}
  \bibinfo{year}{2017}\natexlab{}.
\newblock \showarticletitle{Reachable Sets of Hidden CPS Sensor Attacks:
  Analysis and Synthesis Tools}. In \bibinfo{booktitle}{{\em IFAC 2017 World
  Congress}}.
\newblock


\bibitem[\protect\citeauthoryear{Murguia, van~de Wouw, and Ruths}{Murguia
  et~al\mbox{.}}{2016}]%
        {Carlos_Justin3}
\bibfield{author}{\bibinfo{person}{Carlos Murguia}, \bibinfo{person}{Nathan
  van~de Wouw}, {and} \bibinfo{person}{Justin Ruths}.} \bibinfo{year}{accepred,
  2016}\natexlab{}.
\newblock \showarticletitle{Reachable Sets of Hidden CPS Sensor Attacks:
  Analysis and Synthesis Tools}. In \bibinfo{booktitle}{{\em proceedings of the
  IFAC World Congress}}.
\newblock


\bibitem[\protect\citeauthoryear{Page}{Page}{1954}]%
        {Page}
\bibfield{author}{\bibinfo{person}{E. Page}.} \bibinfo{year}{1954}\natexlab{}.
\newblock \showarticletitle{Continuous Inspection Schemes}.
\newblock \bibinfo{journal}{{\em Biometrika\/}}  \bibinfo{volume}{41}
  (\bibinfo{year}{1954}), \bibinfo{pages}{100--115}.
\newblock


\bibitem[\protect\citeauthoryear{Pasqualetti, D{\"o}rfler, and
  Bullo}{Pasqualetti et~al\mbox{.}}{2013}]%
        {CPSAttackDetection}
\bibfield{author}{\bibinfo{person}{F. Pasqualetti}, \bibinfo{person}{F.
  D{\"o}rfler}, {and} \bibinfo{person}{F. Bullo}.}
  \bibinfo{year}{2013}\natexlab{}.
\newblock \showarticletitle{Attack Detection and Identification in
  Cyber-Physical Systems}.
\newblock \bibinfo{journal}{{\it IEEE Trans. Automat. Control}}
  \bibinfo{volume}{58}, \bibinfo{number}{11} (\bibinfo{date}{Nov}
  \bibinfo{year}{2013}), \bibinfo{pages}{2715--2729}.
\newblock
\showISSN{0018-9286}
\showDOI{%
\url{https://doi.org/10.1109/TAC.2013.2266831}}


\bibitem[\protect\citeauthoryear{Ramadge and Wonham}{Ramadge and
  Wonham}{1987}]%
        {doi:10.1137/0325013}
\bibfield{author}{\bibinfo{person}{P.~J. Ramadge} {and} \bibinfo{person}{W.~M.
  Wonham}.} \bibinfo{year}{1987}\natexlab{}.
\newblock \showarticletitle{Supervisory Control of a Class of Discrete Event
  Processes}.
\newblock \bibinfo{journal}{{\em SIAM Journal on Control and Optimization\/}}
  \bibinfo{volume}{25}, \bibinfo{number}{1} (\bibinfo{year}{1987}),
  \bibinfo{pages}{206--230}.
\newblock
\showDOI{%
\url{https://doi.org/10.1137/0325013}}
\showeprint{https://doi.org/10.1137/0325013}


\bibitem[\protect\citeauthoryear{Rocchetto and Tippenhauer}{Rocchetto and
  Tippenhauer}{2016}]%
        {CPSDolevYao}
\bibfield{author}{\bibinfo{person}{Marco Rocchetto} {and}
  \bibinfo{person}{Nils~Ole Tippenhauer}.} \bibinfo{year}{2016}\natexlab{}.
\newblock \bibinfo{booktitle}{{\em CPDY: Extending the Dolev-Yao Attacker with
  Physical-Layer Interactions}}.
\newblock \bibinfo{publisher}{Springer International Publishing},
  \bibinfo{address}{Cham}, \bibinfo{pages}{175--192}.
\newblock
\showISBNx{978-3-319-47846-3}
\showDOI{%
\url{https://doi.org/10.1007/978-3-319-47846-3\_12}}


\bibitem[\protect\citeauthoryear{Ross}{Ross}{2006}]%
        {Ross}
\bibfield{author}{\bibinfo{person}{M. Ross}.} \bibinfo{year}{2006}\natexlab{}.
\newblock \bibinfo{booktitle}{{\em Introduction to Probability Models, Ninth
  Edition}}.
\newblock \bibinfo{publisher}{Academic Press, Inc.}, \bibinfo{address}{Orlando,
  FL, USA}.
\newblock


\bibitem[\protect\citeauthoryear{Rushby}{Rushby}{1992}]%
        {Rushby92}
\bibfield{author}{\bibinfo{person}{John Rushby}.}
  \bibinfo{year}{1992}\natexlab{}.
\newblock \bibinfo{booktitle}{{\em Noninterference, Transitivity, and
  Channel-Control Security Policies}}.
\newblock \bibinfo{type}{{T}echnical {R}eport}. \bibinfo{institution}{SRI
  International}.
\newblock
\showURL{%
\url{http://www.csl.sri.com/papers/csl-92-2/}}


\bibitem[\protect\citeauthoryear{Sabelfeld and Myers}{Sabelfeld and
  Myers}{2003}]%
        {LanguageBasedInformationFlowSecurity}
\bibfield{author}{\bibinfo{person}{A. Sabelfeld} {and} \bibinfo{person}{A.~C.
  Myers}.} \bibinfo{year}{2003}\natexlab{}.
\newblock \showarticletitle{Language-based information-flow security}.
\newblock \bibinfo{journal}{{\em IEEE Journal on Selected Areas in
  Communications\/}} \bibinfo{volume}{21}, \bibinfo{number}{1}
  (\bibinfo{date}{Jan} \bibinfo{year}{2003}), \bibinfo{pages}{5--19}.
\newblock
\showISSN{0733-8716}
\showDOI{%
\url{https://doi.org/10.1109/JSAC.2002.806121}}


\bibitem[\protect\citeauthoryear{Smith}{Smith}{2011}]%
        {CPSCovertAttacks}
\bibfield{author}{\bibinfo{person}{Roy~S. Smith}.}
  \bibinfo{year}{2011}\natexlab{}.
\newblock \showarticletitle{A Decoupled Feedback Structure for Covertly
  Appropriating Networked Control Systems}.
\newblock \bibinfo{journal}{{\em {IFAC} Proceedings Volumes\/}}
  \bibinfo{volume}{44}, \bibinfo{number}{1} (\bibinfo{year}{2011}),
  \bibinfo{pages}{90 -- 95}.
\newblock
\showISSN{1474-6670}
\showDOI{%
\url{https://doi.org/10.3182/20110828-6-IT-1002.01721}}
\newblock
\shownote{18th {IFAC} World Congress.}


\bibitem[\protect\citeauthoryear{Sutherland}{Sutherland}{1986}]%
        {Nondeducibility}
\bibfield{author}{\bibinfo{person}{D. Sutherland}.}
  \bibinfo{year}{1986}\natexlab{}.
\newblock \showarticletitle{A Model of Information}. In
  \bibinfo{booktitle}{{\em Proccedings of the National Computer Security
  Conference}}. \bibinfo{pages}{175--183}.
\newblock


\bibitem[\protect\citeauthoryear{Urbina, Giraldo, Cardenas, Tippenhauer,
  Valente, Faisal, Ruths, Candell, and Sandberg}{Urbina et~al\mbox{.}}{2016}]%
        {LimitingImpactStealthyAttacks}
\bibfield{author}{\bibinfo{person}{David~I. Urbina}, \bibinfo{person}{Jairo~A.
  Giraldo}, \bibinfo{person}{Alvaro~A. Cardenas}, \bibinfo{person}{Nils~Ole
  Tippenhauer}, \bibinfo{person}{Junia Valente}, \bibinfo{person}{Mustafa
  Faisal}, \bibinfo{person}{Justin Ruths}, \bibinfo{person}{Richard Candell},
  {and} \bibinfo{person}{Henrik Sandberg}.} \bibinfo{year}{2016}\natexlab{}.
\newblock \showarticletitle{Limiting the Impact of Stealthy Attacks on
  Industrial Control Systems}. In \bibinfo{booktitle}{{\em Proceedings of the
  2016 ACM SIGSAC Conference on Computer and Communications Security}} {\em
  (\bibinfo{series}{CCS '16})}. \bibinfo{publisher}{ACM}, \bibinfo{address}{New
  York, NY, USA}, \bibinfo{pages}{1092--1105}.
\newblock
\showISBNx{978-1-4503-4139-4}
\showDOI{%
\url{https://doi.org/10.1145/2976749.2978388}}


\bibitem[\protect\citeauthoryear{van Dobben~de Bruyn}{van Dobben~de
  Bruyn}{1968}]%
        {Dobben}
\bibfield{author}{\bibinfo{person}{C.S. van Dobben~de Bruyn}.}
  \bibinfo{year}{1968}\natexlab{}.
\newblock \bibinfo{booktitle}{{\em Cumulative sum tests : theory and
  practice}}.
\newblock \bibinfo{publisher}{London : Griffin}.
\newblock


\bibitem[\protect\citeauthoryear{Wald}{Wald}{1945}]%
        {wald}
\bibfield{author}{\bibinfo{person}{A. Wald}.} \bibinfo{year}{1945}\natexlab{}.
\newblock \showarticletitle{Sequential Tests of Statistical Hypotheses}.
\newblock \bibinfo{journal}{{\em Ann. Math. Statist.\/}}  \bibinfo{volume}{16}
  (\bibinfo{year}{1945}), \bibinfo{pages}{117--186}.
\newblock


\bibitem[\protect\citeauthoryear{Weerakkody and Sinopoli}{Weerakkody and
  Sinopoli}{2015}]%
        {Weerakkody}
\bibfield{author}{\bibinfo{person}{S. Weerakkody} {and} \bibinfo{person}{B.
  Sinopoli}.} \bibinfo{year}{2015}\natexlab{}.
\newblock \showarticletitle{Detecting integrity attacks on control systems
  using a moving target approach}. In \bibinfo{booktitle}{{\em 2015 54th IEEE
  Conference on Decision and Control (CDC)}}. \bibinfo{pages}{5820--5826}.
\newblock


\bibitem[\protect\citeauthoryear{Weerakkody, Sinopoli, Kar, and
  Datta}{Weerakkody et~al\mbox{.}}{2016}]%
        {IFCPSSec}
\bibfield{author}{\bibinfo{person}{S. Weerakkody}, \bibinfo{person}{B.
  Sinopoli}, \bibinfo{person}{S. Kar}, {and} \bibinfo{person}{A. Datta}.}
  \bibinfo{year}{2016}\natexlab{}.
\newblock \showarticletitle{Information flow for security in control systems}.
  In \bibinfo{booktitle}{{\em 2016 IEEE 55th Conference on Decision and Control
  (CDC)}}. \bibinfo{pages}{5065--5072}.
\newblock
\showDOI{%
\url{https://doi.org/10.1109/CDC.2016.7799044}}


\bibitem[\protect\citeauthoryear{Willsky}{Willsky}{1976}]%
        {Willsky}
\bibfield{author}{\bibinfo{person}{A. Willsky}.}
  \bibinfo{year}{1976}\natexlab{}.
\newblock \showarticletitle{A survey of design methods for failure detection in
  dynamic systems}.
\newblock \bibinfo{journal}{{\em Automatica\/}}  \bibinfo{volume}{12}
  (\bibinfo{year}{1976}), \bibinfo{pages}{601 -- 611}.
\newblock


\end{thebibliography}
\appendix
\section{Appendix}
 {
{
\subsection{Residuals and Hypothesis Testing}

Under certain properties of the matrices $A$ and $C$  (\emph{detectability} \cite{Astrom}), the observer gain $L$ can be designed such that $E[e_k] \rightarrow 0$ as $k \rightarrow \infty$ (where $E[\cdot]$ denotes expectation) in the absence of attacks. Moreover, under the same properties, the covariance matrix  $P_k:= E[e_ke_k^T]$ converges to steady state (in the absence of attacks) in the sense that $\lim_{k \rightarrow \infty} P_k = P$ exists, see \cite{Astrom}. For $\delta^u_k=\delta^y_k=\mathbf{0}$ and appropriately selected $L$, it can be verified that the asymptotic covariance matrix $P = \lim_{k \rightarrow \infty} P_k$ is given by the solution $P$ of the following Lyapunov equation:
\begin{align}
&(F-LC)P(F-LC)^T - P + R_1 + LR_2L^T = \mathbf{0},\label{24}
\end{align}
where $\mathbf{0}$ denotes the zero matrix of appropriate dimensions. It is assumed that the system has reached steady state before an attack occurs. Consider the residual difference equation in (\ref{26}). Then, it can be easily shown that, if there are no attacks, the mean of the residual is
\begin{equation}
E[r_{k+1}] = CE[e_{k+1}] + E[\eta_{k+1}] = \mathbf{0}_{m \times 1},  \label{27} \\
\end{equation}
and its asymptotic covariance matrix is given by
\begin{align}
\Sigma := E[r_{k+1}r_{k+1}^T] &= CPC^T + R_2,\label{28}
\end{align}
where $P$ is the asymptotic covariance matrix of the estimation error $e_k$ solution of (\ref{24}) and $R_2$ is the sensor noise covariance matrix. For this residual, we identify two hypotheses to be tested: $\mathcal{H}_0$ the \emph{normal mode} (no attacks) and $\mathcal{H}_1$ the \emph{faulty mode} (with faults/attacks). Then, we have
\begin{center}
\begin{tabular}{ c c c }
 $\mathcal{H}_0: \left\{
\begin{array}{ll}
E[r_k] = \mathbf{0}, \text{ and}\label{29} \\[.5mm]
E[r_kr_k^T] = \Sigma,
\end{array}
\right.$ & $\mathcal{H}_1: \left\{
\begin{array}{ll}
E[r_k] \neq \mathbf{0}, \text{or}  \label{30} \\[.5mm]
E[r_kr_k^T] \neq \Sigma,
\end{array}
\right.$
\end{tabular}
\end{center}
The objective of hypothesis testing is to distinguish between $\mathcal{H}_0$ and $\mathcal{H}_1$. Several hypothesis testing techniques can be used to examine the residual and subsequently detect faults/attacks. For instance, Sequential Probability Ratio Testing (SPRT) \cite{wald,Willsky}, Cumulative Sum (CUSUM) \cite{Gustafsson,Page}, Genera\-lized Likelihood Ratio (GLR) testing \cite{Basseville}, Compound Scalar Testing (CST) \cite{Gertler}, etc. Each of these techniques has its own advantages and disadvantages depending on the scenario. Arguably, the most utilised one, due to its simplicity, is a particular case of CST, called \emph{bad-data} change detection procedure. 
\subsection{Attack Detection via Bad-Data Procedure}
The input to any detection procedure is a \emph{distance measure} $z_k \in \Real$, with $k \in \Nat$, i.e., a measure of how deviated the estimator is from the sensor measurements. We employ distance measures any time we test to distinguish between $\mathcal{H}_0$ and $\mathcal{H}_1$. Here, we assume there is a dedicated detector on each sensor. Throughout the rest of this section we reserve the index $i$ to denote the sensor/detector, $i\in\mathcal{I}:=\{1,2,\dots,m\}$. Thus, we can partition the attacked output vector as $\bar{y}_k=\text{col}(\bar{y}_{k,1},\ldots,\bar{y}_{k,m})$ where $\bar{y}_{k,i} \in \Real$ denotes the $i$-th entry of $\bar{y}_k \in \Real^m$; then
\begin{equation}
\bar{y}_{k,i}=C_{i}x_k + \eta_{k,i} + \delta_{k,i}^y,\label{31}
\end{equation}
with $C_{i}$ being the $i$-th row of $C$ and $\eta_{k,i}$ and $\delta_{k,i}^y$ denoting the $i$-th entries of $\eta_k$ and $\delta_k^y$, respectively. The bad-data procedure uses the absolute value of the entries of the residual sequence as distance measures:
\begin{equation}
z_{k,i} := |r_{k,i}| = |C_{i}e_k + \eta_{k,i} + \delta_{k,i}^y|.\label{32}
\end{equation}
If there are no attacks,  $r_{k,i}$ follows a normal distribution with zero mean and variance $\sigma_{i}^2$, where $\sigma_{i}^2$ denotes the $i$-th entry of the diagonal of the covariance matrix $\Sigma$. Hence, $\delta_k^y=\delta_k^u=\mathbf{0}$ implies that $z_{k,i} = |r_{k,i}|$ follows a \emph{half-normal distribution} \cite{Ross} with
\begin{equation}
E\big[ \hspace{.25mm}|r_{k,i}| \hspace{.25mm} \big] = \frac{\sqrt{2}}{\sqrt{\pi}}\sigma_{i} \text{ and }  \label{38}
\text{var}\big[  \hspace{.25mm} |r_{k,i}| \hspace{.25mm}  \big] = \sigma_{i}^2 \big( 1 - \frac{2}{\pi} \big).
\end{equation}
Next, having presented the notion of distance measure, we introduce the bad-data procedure.

\noindent\rule{\hsize}{1pt}\vspace{.2mm}
\textbf{Bad-Data Procedure:}
\begin{equation}\label{baddata}
\text{If \ } |r_{k,i}|  > \alpha_i, \hspace{2mm} \tilde{k}_i = k, \text{ \ \ } i \in \mathcal{I}.
\end{equation}
\textbf{Design parameter:} threshold $\alpha_i \in \Real_{>0}$.\\[1mm]
\textbf{Output:} alarm time(s) $\tilde{k}_i$.

\noindent\rule{\hsize}{1pt}\\[.2mm]
The idea is that alarms are triggered if $|r_{k,i}|$ exceeds the threshold $\alpha_i$. The parameter $\alpha_i$ is selected to satisfy a required false alarm rate $\mathcal{A}^*_i$. The occurrence of an alarm in the bad-data when there are no attacks to the CPS is referred to as a false alarm. Operators need to tune this false alarm rate depending on the application. To do this, the thresholds $\alpha_i$ must be selected to fulfill a \emph{desired false alarm rate} $\mathcal{A}^*_i$. Let $\mathcal{A} \in [0,1]$ denote the \emph{false alarm rate} of the bad-data procedure defined as the expected proportion of observations which are false alarms, i.e., $\mathcal{A}_i:=\text{pr}[z_{k,i} \geq \alpha_i]$, where $\text{pr}[\cdot]$ denotes probability, see \cite{Dobben} and \cite{Adams}.
}}

\end{document}